\documentclass[11pt]{article}

\usepackage{geometry} 
\usepackage{amsmath,amsthm,amsfonts,dsfont,mathtools,amssymb}
\usepackage{mathrsfs}
\usepackage{float}
\usepackage{thmtools}
\usepackage{thm-restate}
\usepackage{xcolor}
\usepackage{graphicx}
\usepackage{algorithmic,algorithm}
\usepackage{enumitem}
\setitemize{noitemsep,topsep=5pt,parsep=5pt,partopsep=5pt}

\definecolor{darkgreen}{rgb}{0,0.5,0}
\usepackage{hyperref}
\hypersetup{
    unicode=false,          
    colorlinks=true,        
    linkcolor=violet,          
    citecolor=teal,        
    filecolor=magenta,      
    urlcolor=cyan           
}
\usepackage[capitalize, nameinlink]{cleveref}

\crefname{theorem}{Theorem}{Theorems}
\Crefname{lemma}{Lemma}{Lemmas}
\Crefname{figure}{Figure}{Figures}
\Crefname{claim}{Claim}{Claims}
\Crefname{observation}{Observation}{Observations}

\newtheorem{fact}{Fact}[section]

\newtheorem{theorem}{Theorem}[section]
\newtheorem{lemma}{Lemma}[section]

\newtheorem{corollary}{Corollary}[section]

\newtheorem{observation}{Observation}[section]

\newcommand{\low}{\mathsf{low}}
\newcommand{\high}{\mathsf{high}}
\newcommand{\AAA}{\mathcal{A}}
\newcommand{\VV}{\mathcal{V}}
\newcommand{\HH}{\mathcal{H}}
\newcommand{\QQ}{\mathcal{Q}}
\newcommand{\GGG}{\mathcal{G}}
\newcommand{\PP}{\mathcal{P}}
\newcommand{\OPT}{\mathsf{OPT}}

\newcommand{\Gsp}{G^\diamond}
\newcommand{\Esp}{E^\diamond}
\newcommand{\Vsp}{V^\diamond}
\newcommand{\vstar}{v^\star}
\newcommand{\tmix}{\tau_{\operatorname{mix}}}

\newcommand{\accept}{\mathsf{accept}}
\newcommand{\reject}{\mathsf{reject}}

\newcommand{\outt}{\mathsf{out}}
\newcommand{\inn}{\mathsf{in}}
\newcommand{\vol}{\operatorname{vol}}

\newcommand{\Expect}{\mathbb{E}}
\newcommand{\ID}{\mathsf{ID}}
\newcommand{\poly}{\operatorname{poly}}
\newcommand{\LOCAL}{\mathsf{LOCAL}}
\newcommand{\CONGEST}{\mathsf{CONGEST}}

\newcommand{\dist}{\mathsf{dist}}

\title{Efficient Distributed Decomposition and Routing Algorithms in Minor-Free Networks and Their Applications\footnote{A preliminary version~\cite{chang2023efficient2} of this article was presented at the 42nd ACM Symposium on Principles of Distributed Computing (PODC), June 19--23, 2023, Orlando, Florida}}

\begin{document}

\date{}
\author{Yi-Jun Chang\footnote{National University of Singapore. ORCID: 0000-0002-0109-2432. Email: cyijun@nus.edu.sg}}
\maketitle

\begin{abstract}
In the $\LOCAL$ model of distributed computing, \emph{low-diameter decomposition} is a fundamental tool for algorithm design, as it enables a reduction from general graphs to low-diameter graphs where brute-force information gathering can be performed efficiently.

Chang and Su~[PODC 2022] showed that any high-conductance network excluding a fixed minor contains a high-degree vertex $\vstar$, allowing the entire graph topology to be gathered at $\vstar$ efficiently in the $\CONGEST$ model via \emph{expander routing}. Consequently, in such networks, many problems that admit efficient $\LOCAL$ algorithms via low-diameter decomposition can also be solved efficiently in $\CONGEST$ using \emph{expander decomposition}.

In this work, we present improved decomposition and routing algorithms for networks excluding a fixed minor. We define an \emph{$(\epsilon, D, T)$-decomposition} of a graph $G=(V,E)$ as a partition of $V$ into clusters of diameter at most $D$, with at most $\epsilon |E|$ inter-cluster edges, such that information gathering within each cluster can be completed in $T$ rounds in parallel. 

We show that an $(\epsilon, D, T)$-decomposition with
\[
D = O(\epsilon^{-1}) \quad \text{and} \quad 
T = \min\left\{2^{O\left(\log^2 \frac{1}{\epsilon}\right)} \cdot O(\log \Delta), \ \poly(\epsilon^{-1}, \log \Delta)\right\}
\]
can be computed \emph{deterministically} in
\[
O(\epsilon^{-1} \log^\ast n) + 
\min\left\{2^{O\left(\log^2 \frac{1}{\epsilon}\right)} \cdot O(\log \Delta), \ \poly(\epsilon^{-1}, \log \Delta)\right\}
\]
rounds in the $\CONGEST$ model for networks excluding a fixed minor.

Our algorithm has a wide range of applications, including the following results in $\CONGEST$:
\begin{itemize}
    \item A $(1-\epsilon)$-approximate maximum independent set in networks excluding a fixed minor can be computed deterministically in $O(\epsilon^{-1} \log^\ast n) + \poly(\epsilon^{-1})$ rounds, nearly matching the $\Omega(\epsilon^{-1} \log^\ast n)$ lower bound of Lenzen and Wattenhofer~[DISC 2008].
    \item Property testing of any additive minor-closed property can be performed deterministically in $O(\log n)$ rounds for constant $\epsilon$, or in $O(\epsilon^{-1} \log n) + \poly(\epsilon^{-1})$ rounds for constant $\Delta$, nearly matching the $\Omega(\epsilon^{-1} \log n)$ lower bound of Levi, Medina, and Ron~[PODC 2018].
\end{itemize}
\end{abstract}

\thispagestyle{empty}
\newpage
\thispagestyle{empty}
\tableofcontents
\newpage
\pagenumbering{arabic}

\section{Introduction}\label{sect:intro}

We consider the well-known $\LOCAL$ and $\CONGEST$ models of distributed computing~\cite{Linial92,Peleg00}. In these models, the  communication network is modeled as a graph $G=(V,E)$, where each vertex $v \in V$ corresponds to a computing device and each  edge $e\in E$ corresponds to a communication link.  The topology of $G$ is initially unknown. Each vertex $v \in V$ has a distinct identifier $\ID(v)$ of $O(\log n)$ bits, where $n = |V|$ is the number of vertices in the network.

The communication proceeds in synchronous rounds, and the local computation is free.
In each round of communication, each vertex $v \in V$  sends a message to each of its neighbors $u \in N(v)$, and then  each vertex $v \in V$  receives a message from each of its neighbors $u \in N(v)$. After that,  each vertex $v \in V$ can perform some arbitrary local computation.
In the $\LOCAL$ model, there is no message size constraint. In the $\CONGEST$ model, the size of each message is at most $O(\log n)$ bits. Unless otherwise stated, all algorithms presented in this work apply to the deterministic $\CONGEST$ model.

\paragraph{Minor-free networks.} 
A graph $H$ is a \emph{minor} of a graph $G$ if $H$ can be obtained from $G$ by iteratively doing the following operations: removing vertices, removing edges, and contracting edges. 
A \emph{graph class} $\GGG$ is a set of graphs.
We say that a graph class $\GGG$ is \emph{minor-closed} if it is closed under graph minor operations, that is, if $G \in \GGG$ and $H$ is a minor of $G$, then $H \in \GGG$. 

Many natural graph classes are minor-closed. The list of minor-closed graph classes includes forests, cactus graphs, planar graphs, outer-planar graphs, graphs of genus at most $k$, graphs of treewidth at most $k$,  and graphs of pathwidth at most $k$.

We say that $G$ is \emph{$H$-minor-free} if $H$ is not a minor of $G$. Clearly, if $\GGG$ is minor-closed, then all $G \in \GGG$ are $H$-minor-free, for any choice of $H \notin \GGG$, and such a graph $H$ exists so long as $\GGG$ is not the set of all graphs.  In particular, if we design an algorithm that works for the class of $H$-minor-free graphs, then the algorithm also works for all graphs in $\GGG$. In fact, the graph minor theorem of Robertson and Seymour~\cite{ROBERTSON2004325} implies that for any minor-closed  graph class $\GGG$, there exists a \emph{finite} set of \emph{forbidden minors} $\HH$ such that $G \in \GGG$ if and only if $G$ is $H$-minor-free for all $H \in \HH$. For example, if $\GGG$ is the class of planar graphs, then we may take $\HH=\{K_{3,3}, K_5\}$.

 There is a large body of work designing distributed graph algorithms in $H$-minor-free networks by utilizing their structural properties: distributed approximation~\cite{akhoondian2016local,amiri2019distributed,bonamy2021tight,Czygrinow2006ESA,czygrinow2008fast,czygrinow2014distributed,czygrinow2020distributed,lenzen2013distributed,wawrzyniak2014strengthened}, computation of tree decompositions and its applications~\cite{IzumiSPAA22,li2018distributed}, computation of low-congestion shortcuts and its applications~\cite{ghaffari2021low,ghaffari2016distributed,ghaffari2017near,haeupler2016low,haeupler2016near,haeupler2018minor,haeuplerLi2018disc,haeupler2018round}, planar graph algorithms~\cite{li2019planar,ghaffari2016planar,parter2020distributed}, and local certification~\cite{baterisna2025optimal,DBLP:journals/corr/abs-2108-00059,cook2025tight,esperet2021local,feuilloley2021compact,DBLP:journals/corr/abs-2007-08084,Naor2020soda}.

\paragraph{Low-diameter decomposition.}   An $(\epsilon, D)$ \emph{low-diameter decomposition} of a graph $G=(V,E)$ is 
a collection of sets $\VV = \{V_1, V_2, \ldots, V_k\}$  meeting the following two conditions. 
\begin{itemize}
    \item The collection of sets $\VV$  partitions the vertex set  $V = V_1 \cup V_2 \cup \cdots \cup V_k$ in such a way that the number of inter-cluster edges  is at most   $\epsilon|E|$.
    \item For each $1 \leq i \leq k$, the diameter of the subgraph $G[V_i]$ induced by $V_i$ is at most $D$.
\end{itemize}

In the $\LOCAL$ model, \emph{low-diameter decomposition} is a fundamental tool for algorithm design. It enables a reduction from general graphs to low-diameter graphs, where brute-force information gathering can be performed efficiently. In particular, low-diameter decompositions can be used to construct \emph{network decompositions}, which play a central role in the computational complexity of the $\LOCAL$ model. They also enable the computation of $(1 \pm \epsilon)$-approximate solutions to arbitrary covering and packing integer programs~\cite{chang2023complexity,ghaffari2017complexity}.

In the randomized $\CONGEST$ model, it is known~\cite{EN16,LinialS93,miller2013parallel} that an $(\epsilon, D)$ low-diameter decomposition with $D = O(\epsilon^{-1} \log n)$ can be constructed in $O(\epsilon^{-1} \log n)$ rounds with high probability, where the bound on the number of inter-cluster edges holds in expectation. The first deterministic construction achieving $\poly(\epsilon^{-1}, \log n)$ rounds and diameter $D = \poly(\epsilon^{-1}, \log n)$ in both the $\LOCAL$ and $\CONGEST$ models was given by Rozho\v{n} and Ghaffari~\cite{RozhonG20}. A sequence of subsequent works~\cite{faour2025local,ghaffari2021improved,ghaffari2024near,GGHIR22,rozhovn2022deterministic} further improved both the diameter bound $D$ and the round complexity of deterministic low-diameter decompositions.

\paragraph{Low-diameter decomposition in minor-free networks.} 
Much better constructions of low-diameter decompositions were known for $H$-minor-free networks. It is well-known~\cite{klein1993excluded,Fakcharoenphol2003improved,Ittai2019padded} that for any $H$-minor-free graph, an $(\epsilon, D)$  low-diameter decomposition  with $D = O(\epsilon^{-1})$ exists, and such a decomposition can be used obtain compact routing schemes for $H$-minor-free graphs~\cite{abraham2005compact,abraham2007strong}.

In the $\LOCAL$ model, Czygrinow, Ha{\'n}{\'c}kowiak, and Wawrzyniak~\cite{czygrinow2008fast} showed that an $(\epsilon, D)$  low-diameter decomposition  with $D = \poly(\epsilon^{-1})$ can be constructed in $\poly(\epsilon^{-1}) \cdot O(\log^\ast n)$ rounds deterministically for any $H$-minor-free graph. 
Such a low-diameter decomposition has been used to obtain ultra-efficient approximation algorithms in the $\LOCAL$ model~\cite{amiri2019distributed,czygrinow2020distributed,czygrinow2008fast}. For example, it has been shown that $(1\pm \epsilon)$-approximate solutions for maximum matching, maximum independent set, and minimum dominating set can be computed in $\poly(\epsilon^{-1}) \cdot O(\log^\ast n)$ rounds deterministically for any $H$-minor-free graph in $\LOCAL$~\cite{czygrinow2008fast}. These algorithms inherently require sending large messages, as the algorithms require brute-force information gathering in each low-diameter cluster, so they do not work in $\CONGEST$. 

These components are needed to extend this framework of algorithm design to $\CONGEST$.
\begin{enumerate}
    \item An efficient $\CONGEST$ algorithm for low-diameter decomposition.
    \item Replacing the brute-force information gathering part with an efficient $\CONGEST$ algorithm. 
\end{enumerate}

Levi, Medina, and Ron~\cite{levi2021property} showed a $\poly(\epsilon^{-1}) \cdot O(\log n)$-round algorithm for the problem of {property testing} of {planarity} in   $\CONGEST$ by designing algorithms for the above two components, as follows.
They modified the low-diameter decomposition algorithm of Czygrinow, Ha{\'n}{\'c}kowiak, and Wawrzyniak to make it completes in $\poly(\epsilon^{-1}) \cdot O(\log n)$ rounds in any $H$-minor-free graph in $\CONGEST$. For each low-diameter cluster, instead of using brute-force information gathering, they used the planarity testing algorithm of Ghaffari and Hauepler~\cite{ghaffari2016planar}, which works efficiently for small-diameter graphs in the $\CONGEST$ model.

\paragraph{Expander decomposition and routing.} 
 Roughly speaking, an $(\epsilon,\phi)$ \emph{expander decomposition} of a graph removes at most $\epsilon$ fraction of the edges in such a way that each remaining connected component has  conductance at least $\phi$. In the randomized setting, an expander  decomposition with $\phi= 1/\poly(\epsilon^{-1}, \log n)$ can be computed in $\poly(\epsilon^{-1}, \log n)$ rounds with high probability~\cite{ChangS20}. In the deterministic setting, an expander  decomposition with $\phi= \poly(\epsilon)\cdot n^{o(1)}$ can be computed in $\poly(\epsilon^{-1}) \cdot n^{o(1)}$ rounds~\cite{ChangS20}.

 Recently, Chang and Su~\cite{10.1145/3519270.3538423} showed that any high-conductance $H$-minor-free graph must contain a high-degree vertex $\vstar$, so the entire graph topology can be gathered to $\vstar$ efficiently, even in the $\CONGEST$ model. This information-gathering task can be done using an existing expander routing algorithm  in~\cite{ChangS20,GhaffariKS17,GhaffariL2018}. 
 Consequently, in $H$-minor-free networks, many problems that can be solved efficiently in $\LOCAL$ via low-diameter decomposition can also be solved in $\CONGEST$ via expander decomposition, albeit with a worse round complexity. For example, it was shown in~\cite{10.1145/3519270.3538423} that an  $(1-\epsilon)$-approximate maximum independent set can be computed in $\poly(\epsilon^{-1}) \cdot n^{o(1)}$ rounds deterministically or $\poly(\epsilon^{-1}, \log n)$ rounds with high probability in $\CONGEST$. In the $\LOCAL$ model, the same problem can be solved in   $\poly(\epsilon^{-1}) \cdot O(\log^\ast n)$ rounds deterministically using the low-diameter decomposition algorithm of~\cite{czygrinow2008fast}.

 \subsection{Our contribution}\label{sect:contribution}

 In this work, we show improved decomposition and routing algorithms for $H$-minor-free networks, which allow us to obtain improved algorithms for distributed approximation and property testing.

\paragraph{Our decomposition.}
We define an \emph{$(\epsilon, D, T)$-decomposition} of a graph $G=(V,E)$ by a collection of subsets $\VV = \{V_1, V_2, \ldots, V_k\}$,  an assignment of a leader  $\vstar_S \in V$ to each cluster $S \in \VV$, and a bit string $B_v$ stored locally in each vertex $v \in V$ that is used to run a routing algorithm $\mathcal{A}$, meeting the following requirements:
\begin{itemize}
    \item The collection of sets $\VV$  partitions the vertex set  $V = V_1 \cup V_2 \cup \cdots \cup V_k$ in such a way that the number of inter-cluster edges  is at most   $\epsilon|E|$.
    \item The diameter of the subgraph $G[S]$ induced by each cluster $S \in \VV$ is at most $D$.  
    \item The routing algorithm $\mathcal{A}$ costs at most $T$ rounds and allows each vertex $v \in S$ in each cluster $S \in \VV$ to send $\deg(v)$ messages of $O(\log n)$ bits to vertex $\vstar_S$ in parallel. 
\end{itemize}

The routing algorithm $\mathcal{A}$ allows us to gather the  graph topology of $G[S]$ to $\vstar_S$, for all  $S \in \VV$, in $O(T)$ rounds in parallel. By running $\mathcal{A}$ in reverse, we may let $\vstar_S$ send $\deg(v)$ messages of $O(\log n)$ bits to each  $v \in S$, for all  $S \in \VV$, in $O(T)$ rounds in parallel. We emphasize that the leader $\vstar_S$ does not need to be in the cluster $S$, and we allow multiple clusters to share the same leader.

In $\CONGEST$, using the expander decomposition and routing algorithms in~\cite{ChangS20}, the following results were shown in~\cite{10.1145/3519270.3538423}. In the randomized setting, an $(\epsilon, D, T)$-decomposition with $D = \poly(\epsilon^{-1},\log n)$ and $T = \poly(\epsilon^{-1},\log n)$ can be constructed in $\poly(\epsilon^{-1},\log n)$ rounds with high probability. In the deterministic setting, an $(\epsilon, D, T)$-decomposition with $D = \poly(\epsilon^{-1}) \cdot n^{o(1)}$ and $T = \poly(\epsilon^{-1}) \cdot n^{o(1)}$ can be constructed in $\poly(\epsilon^{-1}) \cdot n^{o(1)}$ rounds.

\paragraph{Our algorithm.}
In this work, we present an improved \emph{deterministic} algorithm for constructing an $(\epsilon, D, T)$-decomposition. 
Our algorithm is \emph{extremely efficient} when $\Delta$ or $\frac{1}{\epsilon}$ is small. For bounded-degree graphs, the round complexity of our algorithm is $O\left(\frac{\log^\ast n }{\epsilon}\right) + \poly\left(\frac{1}{\epsilon}\right)$. For constant $\epsilon$, the round complexity is further reduced to $O\left(\log^\ast n\right)$.

 \begin{restatable}{theorem}{thmmain}\label{lem-algo-main}
    For $\epsilon \in \left(0, \frac{1}{2}\right)$, an $(\epsilon, D, T)$-decomposition of any $H$-minor-free graph $G$ of maximum degree $\Delta$ with $D = O\left(\frac{1}{\epsilon}\right)$
    can be constructed with the following round complexities.
\begin{itemize}
    \item For $T = 2^{O\left(\log^2 \frac{1}{\epsilon}\right)} \cdot  O\left(\log \Delta \right)$,  the round complexity is $O\left(\frac{\log^\ast n }{\epsilon}\right) + 2^{O\left(\log^2 \frac{1}{\epsilon}\right)} \cdot  O\left(\log \Delta \right)$.
    \item For $T = O\left(\frac{\log^5 \Delta \log \frac{1}{\epsilon} +   \log^6 \frac{1}{\epsilon}}{\epsilon^4}\right)$, the round complexity is $O\left(\frac{\log^\ast n }{\epsilon}\right) + O\left(\frac{\log^5 \Delta \log \frac{1}{\epsilon} +   \log^6 \frac{1}{\epsilon}}{\epsilon^5}\right)$.   
\end{itemize}
\end{restatable}

See \cref{sect:routing-algo} for the proof of \cref{lem-algo-main}. See \cref{table-tradeoffs} for the construction time and the routing time $T$ of \cref{lem-algo-main} for different values of $\epsilon$ and $\Delta$.
Our algorithm has a wide range of applications. In $H$-minor-free  networks, many problems that can be solved efficiently in $\LOCAL$ via low-diameter decomposition can also be solved efficiently in $\CONGEST$ via the $(\epsilon, D, T)$-decomposition of \cref{lem-algo-main}. In many cases, we obtain upper bounds that are optimal or nearly optimal.

\begin{table}[ht]
\centering
\begin{tabular}{llll}
{$\mathbf{\Delta}$}                       & $\mathbf{\epsilon}$                    & {\bf Construction time}                                                    &{\bf Routing time}                                       \\ \hline
\multicolumn{1}{|l|}{Constant} & \multicolumn{1}{l|}{Constant} & \multicolumn{1}{l|}{$O(\log^\ast n)$}                                & \multicolumn{1}{l|}{$O(1)$}                        \\ \hline
\multicolumn{1}{|l|}{Constant} & \multicolumn{1}{l|}{Any}      & \multicolumn{1}{l|}{$O(\epsilon^{-1}\log^\ast n) + \poly(\epsilon^{-1})$} & \multicolumn{1}{l|}{$\poly(\epsilon^{-1})$}         \\ \hline
\multicolumn{1}{|l|}{Any}      & \multicolumn{1}{l|}{Constant} & \multicolumn{1}{l|}{$O(\log n)$}                                     & \multicolumn{1}{l|}{$O(\log n)$}                   \\ \hline
\multicolumn{1}{|l|}{Any}      & \multicolumn{1}{l|}{Any}      & \multicolumn{1}{l|}{$\poly(\epsilon^{-1}, \log n)$}                   & \multicolumn{1}{l|}{$\poly(\epsilon^{-1}, \log n)$} \\ \hline
\end{tabular}
\caption{The complexities of $(\epsilon, D, T)$-decompositions with $D=O(\epsilon^{-1})$ in \cref{lem-algo-main}.}
\label{table-tradeoffs}
\end{table}

\paragraph{Distributed approximation.}
\cref{lem-algo-main} can be used in designing approximation algorithms in the following way. In an $(\epsilon, D, T)$-decomposition $\VV$, the routing algorithm $\AAA$ associated with the decomposition allows the leader $\vstar_S$ of each cluster $S \in \VV$ to gather the entire graph topology of $G[S]$ in $O(T)$ rounds, so $\vstar_S$ can perform any arbitrary local computation for the graph $G[S]$. Intuitively, if the optimization problem under consideration has the property that ignoring all the inter-cluster edges can only affect the approximation ratio by a factor of $O(\epsilon)$, then a $(1\pm O(\epsilon))$-approximate solution can be computed by combining optimal solutions of $G[S]$ for all $S \in \VV$, where such an optimal solution can be computed by  $\vstar_S$ locally and announced to all vertices in $S$ via the routing algorithm $\AAA$ in $O(T)$ rounds. For example, based on this approach, a $(1-\epsilon)$ maximum cut of any $H$-minor-free  network can be computed with  round complexity that is linear in the construction time and  the routing time $T$ of the $(\epsilon, D, T)$-decomposition $\VV$.

For certain combinatorial optimization problems, we can further improve the round complexity to  $\poly(\epsilon^{-1}) \cdot O(\log^\ast n)$ using the \emph{bounded-degree sparsifiers} introduced by Solomon in~\cite{solomon:LIPIcs:2018:8364}.
It was shown in~\cite{solomon:LIPIcs:2018:8364} that for maximum matching, maximum independent set, and minimum vertex cover, there exists a deterministic \emph{one-round} reduction that reduces the problem of finding a $(1\pm \epsilon)$-approximate solution in a bounded-arboricity graph to the same problem in a subgraph with $\Delta=O(\epsilon^{-1})$. Therefore, for these problems, we may focus on the case of $\Delta=O(\epsilon^{-1})$, so the approach mentioned above allows us to find $(1\pm \epsilon)$-approximate solutions in $\poly(\epsilon^{-1}) \cdot O(\log^\ast n)$ rounds.
In particular, we show that a $(1-\epsilon)$-approximate maximum independent set in any $H$-minor-free graph can be computed \emph{deterministically} in \[O(\epsilon^{-1}\log^\ast n) + \poly(\epsilon^{-1})\] rounds, nearly matching the  $\Omega(\epsilon^{-1}\log^\ast n)$ lower bound of Lenzen and Wattenhofer~\cite{LenzenW08}.
See \cref{sect:application-apx} for details.

Our algorithm improves upon the previous $(1-\epsilon)$-approximate maximum independent set algorithm in~\cite{10.1145/3519270.3538423}, which costs $\poly(\epsilon^{-1}) \cdot n^{o(1)}$ rounds deterministically or $\poly(\epsilon^{-1}, \log n)$ rounds with high probability in $\CONGEST$. Before this work, the round complexity $\poly(\epsilon^{-1}) \cdot O(\log^\ast n)$ was only known to be attainable in the $\LOCAL$ model~\cite{czygrinow2008fast}.

\paragraph{Distributed property testing.}  A \emph{graph property} $\mathcal{P}$ is a set of graphs. We say that a graph $G=(V,E)$ is \emph{$\epsilon$-far} from having the property $\PP$ if we need to insert or delete at least $\epsilon|E|$ edges to obtain property $\PP$. 
We say that a deterministic distributed algorithm $\AAA$ is a  \emph{property testing} algorithm for property $\PP$ if the following holds when we run $\AAA$ on any graph $G=(V,E)$.
\begin{itemize}
    \item If $G$ has property $\PP$, then all vertices $v\in V$ output $\accept$.
    \item If $G$ is $\epsilon$-far from having property $\PP$, then at least one vertex $v\in V$ outputs $\reject$.
\end{itemize}

\cref{lem-algo-main}  can be used to design efficient property testing algorithms for additive and minor-closed graph properties.
We say that a graph property $\mathcal{P}$ is \emph{minor-closed} if $G \in \mathcal{P}$ implies $H \in \mathcal{P}$ for each minor $H$ of $G$. We say that  a graph property $\mathcal{P}$ is \emph{additive} if  $\mathcal{P}$ is closed under disjoint union of graphs, that is, if $G_1 \in \mathcal{P}$ and $G_2 \in \mathcal{P}$, then the disjoint union of $G_1$ and $G_2$ is also in $\mathcal{P}$.
For example, the set of planar graphs is a graph property that is both additive and minor-closed.

We develop an error detection algorithm, which uses the \emph{forests decomposition} algorithm of Barenboim and Elkin~\cite{BE10}, that can detect an error  whenever \cref{lem-algo-main} produces an incorrect output due to the fact that $G$ is not $H$-minor-free. Combining the error detection algorithm with \cref{lem-algo-main}, we show that any additive and minor-closed graph property can be tested   
\emph{deterministically} in \[O(\epsilon^{-1}\log n) + \min\left\{2^{O\left(\log^2 \frac{1}{\epsilon}\right)} \cdot  O\left(\log \Delta \right), \poly(\epsilon^{-1},\log \Delta)\right\}\] rounds.
If $\epsilon$ is a constant, the round complexity is $O(\log n)$.  If $\Delta$ is a constant, then the round complexity is  $O(\epsilon^{-1}\log n) + \poly(\epsilon^{-1})$. These round complexities nearly match the  $\Omega(\epsilon^{-1}\log n)$ lower bound of Levi, Medina, and Ron~\cite{levi2021property}. See \cref{sect:application-test} for details.

Our algorithm improves upon the previous property testing algorithm for additive and minor-closed graph properties in~\cite{10.1145/3519270.3538423}, which costs $\poly(\epsilon^{-1}) \cdot n^{o(1)}$ rounds deterministically or $\poly(\epsilon^{-1}, \log n)$ rounds with high probability in $\CONGEST$. Previously, the round complexity $\poly(\epsilon^{-1}) \cdot O(\log n)$ in $\CONGEST$ was only known to be achievable for property testing of planarity~\cite{levi2021property}.

 \subsection{Technical overview}\label{sect:tech-overview}
  
 All existing distributed expander decomposition algorithms~\cite{Chang2021triangleJACM,ChangS20} are based on a \emph{top-down} approach that iteratively finds sparse cuts and low-diameter decompositions until each remaining connected component has high conductance.
 With the existing techniques, this approach inherently needs $\poly(\epsilon^{-1}, \log n)$ rounds, as this is the best-known upper bound~\cite{GGHIR22,miller2013parallel,RozhonG20} for computing a low-diameter decomposition in $\CONGEST$.

 In this work, we consider a \emph{bottom-up} approach to building our decomposition. We start with the trivial clustering where each vertex is a cluster, and then we iteratively merge the clusters. This approach is commonly used in low-diameter decomposition algorithms~\cite{czygrinow2008fast,RozhonG20}.
 
 In the $\LOCAL$ model, for any given clustering $\VV = \{V_1, V_2, \ldots, V_k\}$ that partitions the vertex set  $V$ of an $H$-minor-free graph $G=(V,E)$, the \emph{heavy-stars} algorithm of Czygrinow, Ha{\'n}{\'c}kowiak, and Wawrzyniak~\cite{czygrinow2008fast}  identifies a number of disjoint stars in the cluster graph such that once we merge these stars, the number of inter-cluster edges is reduced by a constant factor. This merging algorithm was also utilized in the subsequent low-diameter decomposition algorithms of Levi, Medina, and Ron~\cite{levi2021property}.
 
 If the goal is to compute a \emph{low-diameter decomposition} in the $\LOCAL$ model, then all we need to do is to repeatedly run the heavy-stars algorithm for $O(\epsilon^{-1})$ iterations. This ensures that the number of inter-cluster edges becomes at most $\epsilon$ fraction.
 To apply this approach to finding an \emph{expander decomposition} in $\CONGEST$, we need to deal with a few issues.
 
 The first issue needed to be resolved is that the heavy-stars algorithm requires each cluster $S \in \VV$ to identify a neighboring cluster $S' \in \VV$ such that the number of edges crossing $S$ and $S'$ is maximized. This step requires unbounded message size, and it is the only reason that the algorithm of  Czygrinow, Ha{\'n}{\'c}kowiak, and Wawrzyniak works in the $\LOCAL$ model and not the $\CONGEST$ model. Levi, Medina, and Ron~\cite{levi2021property} considered two approaches to deal with this issue. The first approach is to spend logarithmic rounds to find a \emph{forest decomposition}. The second approach is to select $S'$ randomly. Both approaches are not applicable to our setting as we aim for a sublogarithmic-round deterministic algorithm. 
 
 This issue can be resolved naturally if each cluster $S \in \VV$ induces a high-conductance subgraph $G[S]$, assuming that we have an efficient \emph{information-gathering} algorithm to gather all the necessary information to a high-degree vertex $\vstar \in S$, so $\vstar$ can use its local computation to determine a neighboring cluster $S' \in \VV$ such that the number of edges crossing $S$ and $S'$ is maximized.

 In this work, we design two efficient deterministic information-gathering algorithms in high-conductance $H$-minor-free graphs, based on derandomizing random walks with limited independence and a load-balancing algorithm of~\cite{GhoshLMMPRRTZ99}. Interestingly, despite that the two algorithms are based on completely different approaches, they have very similar round complexities. See \cref{sect:routing} for details.

We still need to address the issue that the algorithm of Czygrinow, Ha{\'n}{\'c}kowiak, and Wawrzyniak produces clusters with small diameter but not necessarily high conductance. To overcome this, we observe that if the number of edges between two clusters $S$ and $S'$ considered for merging is too small, we can simply refrain from merging them while still achieving the desired reduction in the total number of inter-cluster edges.

Building on this observation, together with additional ideas, we modify the merging procedure to ensure that each merging step degrades conductance by at most a factor of $O(\epsilon^{-1})$, at the expense of allowing slight overlap between clusters.
 
 Since there are $O\left(\log \frac{1}{\epsilon}\right)$ iterations in total, in the end, all clusters have conductance $\phi = 2^{-O\left(\log^2 \frac{1}{\epsilon}\right)}$, which is good in the regime that $\epsilon$ is not too small. We will show that the modified algorithm can be implemented to run in \[2^{O\left(\log^2 \frac{1}{\epsilon}\right)} \cdot  O(\log \Delta)  \cdot \left(O(\log^3 \Delta) +O(\log^\ast n)\right).\] rounds deterministically in $\CONGEST$, and the algorithm outputs a variant of expander decomposition  with $\phi = 2^{-O\left(\log^2 \frac{1}{\epsilon}\right)}$ that allows the clusters to slightly overlap. See \cref{sect:existence-variant} for details.

The above construction can be further improved by employing an information-gathering procedure that enables each cluster $S \in \VV$ to locally compute the \emph{best-known} partition of $S$ into small-diameter sub-clusters that allows efficient information gathering for all sub-clusters in parallel. For example, since any $H$-minor-free graph admits an expander decomposition with conductance $\phi = \Omega\left(\frac{\epsilon}{\log \frac{1}{\epsilon} + \log \Delta}\right)$, this approach allows us to compute an expander decomposition matching this bound in the $\CONGEST$ model.

\cref{lem-algo-main} follows by combining the ideas discussed above. See \cref{sect:decomposition} for details.

\paragraph{Remark.} The decomposition and routing algorithms of this work utilize the following properties of the class $\GGG$ of $H$-minor-free graphs.
\begin{enumerate}
    \item \label{xx1} Each graph $G \in \GGG$ has bounded arboricity.
    \item \label{xx2} Each graph $G \in \GGG$ satisfies $\Delta = \Omega(\phi^2 n)$, where $\Delta$, $\phi$, and $n$ are the maximum degree, the conductance, and the number of vertices of $G$, respectively.
    \item \label{xx4} Graph class $\GGG$ is closed under the contraction operation.
    \item \label{xx3} Graph class $\GGG$ is closed under the subgraph operation.
\end{enumerate}
 \cref{xx1,xx4} allow us to show that the cluster graph for any partition of the vertex set $V = V_1 \cup V_2 \cup \cdots \cup V_k$ has bounded arboricity. This property is needed in the heavy-stars algorithm of Czygrinow, Ha{\'n}{\'c}kowiak, and Wawrzyniak~\cite{czygrinow2008fast}.  \cref{xx2,xx3} allow us to show that each cluster $S$ in an expander decomposition of $G$ contains a high-degree vertex. This property ensures that  brute-force information gathering in $S$ can be done efficiently in $\CONGEST$. \cref{xx2} was proved in~\cite{10.1145/3519270.3538423} by showing that $H$-minor-free graphs admit small balanced edge separators.

 \subsection{Organization} \label{sect:organization}

In \cref{sect:routing}, we present our distributed information-gathering algorithms in $H$-minor-free high-conductance graphs. 
In \cref{sect:existential}, we present simple existential bounds for expander decompositions in $H$-minor-free graphs.
In \cref{sect:existence-variant}, we show that in $H$-minor-free graphs, the conductance bound $\phi = 2^{-O\left(\log^2 \frac{1}{\epsilon}\right)}$ can be achieved for a variant of expander decomposition that allows clusters to slightly overlap and demonstrate an efficient distributed algorithm constructing such a decomposition.
In \cref{sect:decomposition}, we combine the results in previous sections to design an efficient algorithm for $(\epsilon, D, T)$-decomposition. 
In \cref{sect:application}, we demonstrate some applications of our $(\epsilon, D, T)$-decomposition. 
In \cref{sect:conclusions}, we conclude our work and discuss some open questions.

\section{Distributed information gathering}\label{sect:routing}
We begin with some basic graph terminology. For each vertex $v$, we write $N(v)$ to denote the set of neighbors of $v$. We write $\dist(u,v)$ to denote the distance between $u$ and $v$. 
For any two vertex subsets $A \subseteq V$ and $B \subseteq V$, we write $E(A,B)$ to denote the set of all edges $e=\{u,v\}$ with $u \in A$ and $v \in B$.

\paragraph{Conductance.}
For any subset $S \subseteq V$, we write $\partial(S)=E(S, V \setminus S)$. 
We define the \emph{volume} of a subset $S \subseteq V$ as $\vol(S) = \sum_{v\in S} \deg(v)$, where the degree is measured with respect to the underlying graph $G$ and not the subgraph $G[S]$ induced by $S$.
For any subset $S \subseteq V$ with $S \neq \emptyset$ and $S \neq V$, we define the \emph{conductance} of $S$ as follows.
\[\Phi(S) = \frac{|\partial(S)|}{\min\{\vol(S), \vol(V \setminus S)\}}.\]
We define the \emph{conductance}  $\Phi(G)$ of a graph $G$ as the minimum value of $\Phi(S)$ over all subsets $S \subseteq V$ with $S \neq \emptyset$ and $S \neq V$.  We say that $G$ is a \emph{$\phi$-expander} if $G$ satisfies $\Phi(G) \geq \phi$.

\paragraph{Sparsity.} For any subset $S \subseteq V$ with $S \neq \emptyset$ and $S \neq V$, we define the  \emph{sparsity} of  $S$ as follows. \[\Psi(S) = \frac{|\partial(S)|}{\min\{ |S|, |V \setminus S|\}}.\]
Similarly, the \emph{sparsity} $\Psi(G)$ of a graph $G$ is the minimum value of $\Psi(S)$ over all subsets $S \subseteq V$ with $S \neq \emptyset$ and $S \neq V$.
To put it another way, sparsity is a variant of conductance that measures the size of a vertex subset $A$ by its cardinality $|A|$ instead of its volume $\vol(A)$. As a result, we always have $\Phi(S) \leq \Psi(S) \leq \Delta \cdot \Phi(S)$ and $\Phi(G) \leq \Psi(G) \leq \Delta \cdot \Phi(G)$. Sparsity is also commonly known as \emph{edge expansion}.

\paragraph{Expander split.}
The  \emph{expander split} $\Gsp=(\Vsp, \Esp)$ of $G=(V,E)$ is constructed as follows.
\begin{itemize}
    \item For each $v \in V$, construct a $\deg(v)$-vertex graph $X_v$ with $\Delta(X_v) = \Theta(1)$ and $\Phi(X_v) = \Theta(1)$.
    \item Each vertex $v \in V$ arbitrarily orders its incident edges.
    For each edge $e=\{u,v\} \in E$, add an edge between the $r_u(e)$th vertex of $X_u$ and the $r_v(e)$th vertex of $X_v$, where $r_u(e)$ is the rank of $e$ for $u$ and $r_v(e)$ is the rank of $e$ for $v$.
\end{itemize}
Refer to~\cite{chuzhoy2019deterministic,ChangS20} for properties of expander split. The only property that we need here is that $\Psi(\Gsp)$ and $\Phi(G)$ are within a constant factor of each other, see~\cite[Lemma C.2]{ChangS20}.

\paragraph{Information gathering.}
The goal of this section is to solve the following information-gathering task in a $\phi$-expander $G=(V,E)$.
Let $\vstar \in V$ be chosen such that $\deg(\vstar) = \Delta$, where $\Delta$ is the maximum degree of $G$.  Each vertex $v \in V$ wants to deliver $\deg(v)$ messages of $O(\log n)$ bits to $\vstar$. Our goal is to deliver at least $1-f$ fraction of these messages.

In \cref{sect:balance}, we solve this problem using a load-balancing algorithm of Ghosh~et~al.~\cite{GhoshLMMPRRTZ99}. In \cref{sect:walks}, we consider the setting where the entire graph topology of $G$ is already known to some vertex $v'$, and we present a different algorithm based on derandomizing random walks. The round complexities of both algorithms are polynomial in $\frac{|E|}{\Delta}$, $\phi^{-1}$, $\log |E|$, and $\log f^{-1}$.

\subsection{Load balancing}\label{sect:balance}

Consider a setting in which each vertex $v$ holds a number of tokens, and the goal is to balance the load across all vertices. Ghosh~et~al.~\cite{GhoshLMMPRRTZ99} study the following natural load-balancing algorithm: at each step, every vertex $v$ sends one token to each neighbor $u \in N(v)$ whose load at the beginning of the step is at least $2\Delta + 1$ smaller than that of $v$. The threshold $2\Delta + 1$ ensures that whenever $v$ sends a token to $u$, vertex $v$ still holds more tokens than $u$ after the step.

We define the \emph{total imbalance} as the maximum, over all vertices $v \in V$, of the absolute difference between the load of $v$ and the average load per vertex. The following lemma was established in~\cite{GhoshLMMPRRTZ99}.

\begin{lemma}[\cite{GhoshLMMPRRTZ99}]\label{lem:balance}
Suppose $G=(V,E)$ has maximum degree $\Delta$ and sparsity $\psi$.
If the total imbalance is $M$ at the beginning, then $O(M \psi^{-1})$ steps of the load-balancing algorithm is sufficient to reduce the total imbalance to $O(\Delta^2 \psi^{-1} \log |V|)$. 
\end{lemma}

A direct application of \cref{lem:balance} to the underlying communication network $G$ is inefficient, as the maximum degree $\Delta$ may be large. To address this, we instead consider the expander split $\Gsp = (\Vsp, \Esp)$ of $G = (V,E)$. In the distributed setting, $\Gsp$ can be simulated within $G$ at no additional cost, allowing us to run the load-balancing algorithm on $\Gsp$ instead.

In the following lemma, we show that this approach yields an efficient information-gathering algorithm. In particular, setting $f = \frac{1}{2|E| + 1}$ ensures that all $2|E|$ messages are delivered.

\begin{lemma}\label{lem:gathering-1}
Let $G=(V,E)$ be an $\phi$-expander, and select $\vstar \in V$  such that $\deg(\vstar)$ equals the maximum degree $\Delta$ of $G$. Suppose each vertex $v \in V$ wants to 
send $\deg(v)$ messages of $O(\log n)$ bits  to $\vstar$.
For any $0 < f < \frac{1}{2}$, there is an algorithm that delivers at least $1-f$ fraction of these $2|E|$ messages to $\vstar$ in $O(\phi^{-2} \Delta^{-1} |E|  \log |E| \log^2 f^{-1})$ rounds. 
\end{lemma}
\begin{proof}
Let $C > 0$ be a constant such that $C \phi^{-1} \log |E|$ is the total imbalance upper bound guaranteed by \cref{lem:balance} when we apply the load-balancing algorithm to $\Gsp=(\Vsp, \Esp)$. Observe that $\Gsp$ has maximum degree $\Theta(1)$ and sparsity $\Theta(\phi)$, and the number of vertices in $\Gsp$ is $|\Vsp| = 2|E|$, so $\log |\Vsp| = \Theta(\log |E|)$.

\paragraph{Delivering a fraction of the messages.}
The $\Delta$ messages in $\vstar$ are already at the destination, so we ignore them in the subsequent discussion.
As a warm-up, we first show how we can deliver at least $\frac{\Delta}{8 |E|}$ fraction of the messages in $V \setminus \{\vstar\}$ in $O(\phi^{-2} \log |E|)$ rounds.
For each $v \in V \setminus \{\vstar\}$, we associate each message of $v$ with a distinct vertex of $X_v$ in the expander split $\Gsp=(\Vsp, \Esp)$, so each $u \in \Vsp \setminus X_{\vstar}$ holds exactly one message. For each $u \in \Vsp$, we create $4 C \phi^{-1} \log |E|$ tokens, where each token contains the message that $u$ holds. The average load per 
 vertex in $\Vsp$ is 
 \[L = \frac{|\Vsp \setminus X_{\vstar}|}{|\Vsp|} \cdot  4 C \phi^{-1} \log |E| = \frac{2|E| - \Delta}{2|E|} \cdot  4 C \phi^{-1} \log |E| \geq  2 C \phi^{-1} \log |E|.\]
 By \cref{lem:balance}, we may reduce the total imbalance to at most $C \phi^{-1} \log |E|$ by running the load-balancing algorithm for  $O(\phi^{-2} \log |E|)$ steps, as the initial total imbalance $M$ is upper bounded by $4C\phi^{-1} \log |E|$. At the end of the load-balancing algorithm, each $u \in X_{\vstar}$ holds at least $L - C \phi^{-1} \log |E| \geq C \phi^{-1} \log |E|$ tokens, meaning that at least
 \[\frac{|X_{\vstar}|}{4|\Vsp|} = \frac{\Delta}{8 |E|}\]
 fraction of the messages are delivered to $\vstar$. By running the algorithm in reverse, each vertex $v \in V$ can learn which of its messages are successfully delivered to $\vstar$.

\paragraph{The remaining messages.}
At first glance, it may seem that repeating the above algorithm for the remaining messages over $O(\Delta^{-1} |E| \log f^{-1})$ iterations should suffice to deliver a $(1-f)$-fraction of the messages to $\vstar$. However, there is a subtle obstacle to this approach.

Suppose $h$ denotes the current fraction of messages that remain undelivered. To guarantee that the average load $L$ per vertex in $\Vsp$ satisfies
\[
L \ge 2C\phi^{-1}\log |E|,
\]
we would need to create $\Omega(h^{-1}\phi^{-1}\log |E|)$ tokens per message. If all of these tokens are created, then the initial total imbalance becomes
\[
M = O(h^{-1}\phi^{-1}\log |E|),
\]
which can be very large when $h$ is small. This is problematic because the load-balancing algorithm requires $O(\phi^{-1} M)$ steps.

\paragraph{Token splitting.}
 To resolve this issue, we create these tokens in the following way. Initially, each $u \in \Vsp$ that holds a message creates only $4 C \phi^{-1} \log |E|$ tokens for its message, and then we run the load-balancing algorithm for $O(\phi^{-2} \log |E|)$ steps to reduce the total imbalance to at most $C \phi^{-1} \log |E|$. After that, we split each token into two tokens, and then we run the load-balancing algorithm for $O(\phi^{-2} \log |E|)$ steps to reduce the total imbalance to at most $C \phi^{-1} \log |E|$ again. We repeat the token splitting and the load-balancing algorithm until the 
 average load $L$ per vertex in $\Vsp$ satisfies $L \geq  2 C \phi^{-1} \log |E|$. We can deduce that at this moment, the total number of tokens is at most $|\Vsp| \cdot 4 C \phi^{-1} \log |E|$. If the total number of tokens exceeds that number, then we already have $L \geq  2 C \phi^{-1} \log |E|$ before the last token splitting, so we should stop the algorithm at that time. 

 Similarly, at the end of the above procedure, each $u \in X_{\vstar}$ holds at least $L - C \phi^{-1} \log |E| \geq C \phi^{-1} \log |E|$ tokens. Since the total number of tokens is at most $|\Vsp| \cdot 4 C \phi^{-1} \log |E|$, we infer that at least
 \[\frac{|X_{\vstar}|}{4|\Vsp|} = \frac{\Delta}{8 |E|}\]
 fraction of the $k$ messages are delivered to $\vstar$. Similarly, by running the above procedure in reverse, each vertex $v \in V$ can learn which of its messages are successfully delivered to $\vstar$.

 \paragraph{Round complexity.}
 We may assume that $k > f \cdot 2|E|$, since otherwise we have already delivered $1-f$ fraction of the messages.
 Since the number of tokens initially is $k \cdot 4 C \phi^{-1} \log |E|$, the number of repetitions of the  token splitting and the load-balancing algorithm needed is 
 $O(\log |\Vsp| - \log k) = O(\log f^{-1})$, as $|\Vsp| = 2|E|$ and $k > f \cdot 2|E|$. The round complexity of the load-balancing algorithm is $O(\phi^{-2} \log |E|)$. Checking whether the bound $L \geq  2 C \phi^{-1} \log |E|$ is met costs $O(D)$ rounds, where $D = O(\phi^{-1} \log |V|)$ is the diameter of the graph $G$. Therefore, we conclude that in $O(\phi^{-2} \log |E| \log f^{-1})$ rounds we may deliver $\frac{\Delta}{8 |E|}$ fraction of the remaining $k$ messages to $\vstar$, regardless of the value of $k > f \cdot 2|E|$. 

 Initially, we have  $k = \sum_{v \in V \setminus \{\vstar\}} \deg(v) < 2|E|$ messages needed to be delivered to $\vstar$.
 By repeating the above procedure for  $t$ iterations to the remaining messages, we will be able to reduce the number of remaining messages to at most 
 \[\left(1 - \frac{\Delta}{8 |E|}\right)^t \cdot 2|E|.\]
 By setting  $t = \Theta(\Delta^{-1} |E| \log f^{-1})$ with a sufficiently large hidden constant, the number of remaining messages will be at most $f \cdot 2|E|$, meaning that $1-f$ fraction of the messages have been sent to $\vstar$. The overall round complexity is $t \cdot 
 O(\phi^{-2} \log |E| \log f^{-1}) = O(\Delta^{-1} |E| \phi^{-2} \log |E| \log^2 f^{-1})$.
\end{proof}
 
\subsection{Random walks}\label{sect:walks} 
 We consider the setting where the entire graph topology of $G$ is known to some vertex $v'$. We will show that in this case, $v'$ can locally compute an efficient routing schedule for the  information-gathering task of \cref{lem:gathering-1} and encode the routing schedule with a small number of bits, so $v'$ can afford to broadcast the routing schedule to all vertices in $G$, and then the vertices in $G$ can run the routing algorithm according to the routing schedule prepared by $v'$. The way $v'$ computes the routing schedule is by  derandomizing random walks.

\paragraph{Lazy random walks.} In each step of a \emph{lazy random walk}, with probability $1/2$, we stay at the current vertex, and with probability $1/2$, we move to a uniform random neighbor. More formally, we let $V=\{v_1, v_2, \ldots, v_{|V|}\}$ and define the \emph{adjacency matrix} $A$ by the $|V| \times |V|$ matrix where $A_{i,j}$ is the number of edges between $v_i$ and $v_j$. If $G$ is a simple graph, then $A_{i,j} \in \{0,1\}$ indicates whether $\{v_i, v_j\} \in E$. For the case of $i = j$, $A_{i,i}$ indicates the number of self-loops at $v_i$. Let $D$ be the matrix such that $D_{i,j} = 0$ when $i \neq j$ and $D_{i,i} = 1/\deg(v_i)$, where each self-loop at $v_i$ contributes one to the calculation of $\deg(v_i)$. Consider any probability distribution $p=(p_1, p_2, \ldots, p_{|V|})^\top$ over $V$ such that $p_i$ indicates the probability that the lazy random walk is currently at $v_i$. Then the probability distribution for the next step of the lazy random walk is given by the formula:
\[\frac{1}{2} \cdot p + \frac{1}{2} \cdot ADp.\]

\paragraph{Mixing time.}
For any integer $t \geq 0$, for any two vertices $v \in V$ and $u \in V$, we write $p_{v}^t(u)$ to denote the probability that the lazy random walk starting from $v$ is at $u$ after $t$ steps. 
Following~\cite{GhaffariKS17}, we define the \emph{mixing time} $\tmix(G)$ as the smallest integer $t$ such that the following holds for all $u$ and $v$:
\[\left|p_v^{t}(u) - \frac{\deg(u)}{\sum_{w \in V} \deg(w)}\right| \leq \frac{1}{|V|} \cdot \frac{\deg(u)}{\sum_{w \in V} \deg(w)}.\]
It is well-known~\cite{GhaffariKS17,JerrumS89} that if $G$ is a $\phi$-expander, then 
\[\tmix(G) = O(\phi^{-2} \log |V|).\]

\paragraph{Limited independence.} A set $A$ of random variables is \emph{$k$-wise independent} if any $k$ random variables in $A$ are mutually independent.  It is well-known~\cite{AlonSpencer} that a $k$-wise independent collection $X=\{a_1, a_2, \ldots, a_{s}\}$ of $s$ binary random variables with $\Pr[a_i = 0] = \Pr[a_i = 1] = \frac{1}{2}$ can be constructed from a   collection $B=\{b_1, b_2, \ldots, b_t\}$ of $t = O(k \log s)$ mutually independent binary random variables with $\Pr[b_i = 0] = \Pr[b_i = 1] = \frac{1}{2}$.

Suppose $X=\sum_{i=1}^s x_i$ such that $x_1, x_2, \ldots, x_s$ are $\lceil\mu \delta \rceil$-wise independent random variables taking values from $\{0,1\}$, where $0 < \delta < 1$ and $\mu = \Expect[X]$.
We have the following tail bound~\cite{SchmidtSS95}. 
\[
\Pr\left[X \geq (1+\delta)\mu\right] \leq  
\left(\frac{e^{\delta}}{(1+\delta)^{(1+\delta)}}\right)^{\mu}\leq \begin{cases}
    e^{-\frac{\mu \delta^2}{3}}, & \ \  \text{if } \delta < 1,\\
    e^{-\frac{\mu \delta}{3}}, & \ \  \text{if } \delta \geq 1.
\end{cases}
\]

\paragraph{Lazy random walks with limited independence.} Recall that our goal is to solve the following information-gathering problem. Let $G=(V,E)$ be an $\phi$-expander, and select $\vstar \in V$  such that $\deg(\vstar)$ equals the maximum degree $\Delta$ of $G$. Suppose each vertex $v \in V$ wants to 
send $\deg(v)$ messages of $O(\log n)$ bits  to $\vstar$. Our goal is to design an efficient distributed algorithm that delivers at least $1-f$ fraction of these  messages to $\vstar$. 

Similar to \cref{sect:balance}, we will consider the expander split $\Gsp=(\Vsp, \Esp)$. We will add self-loops to each vertex in $\Vsp$ in such a way that all vertices have the same degree $d=O(1)$ that is an integer multiple of two. We write  $\widetilde{\Gsp}=(\Vsp, \widetilde{\Esp})$ to denote the resulting graph. The conductance of $\widetilde{\Gsp}$ is within a constant factor of the conductance of $\Gsp$, which is within a constant factor of $\Phi(G)  \geq \phi$. Therefore, we have \[\tmix(\widetilde{\Gsp}) = O(\phi^{-2} \cdot \log |E|).\]

Similar to the proof of \cref{lem:gathering-1}, for each $v \in V$, we associate each message of $v$ to a distinct vertex of $X_v$ in the expander split, so each $u \in \Vsp$ holds exactly one message. For each $u \in \Vsp$, we initiate \[r = C \cdot \left( \frac{|\Vsp|}{|X_{\vstar}|} \cdot \log \frac{1}{f} + \log \tmix(\widetilde{\Gsp})\right) = O\left(\frac{|E|}{\Delta} \cdot  \log \frac{1}{f} + \log \phi^{-1} + \log \log |E|\right)\] lazy random walks in $\widetilde{\Gsp}$, where each lazy random walk contains the message that $u$ holds. Here $C > 0$ is selected as a large enough constant to make all the subsequent proofs work.

We will run each lazy random walk for $\tmix(\widetilde{\Gsp})$ steps. Since $\widetilde{\Gsp}$ is $d$-regular where $d$ is an integer multiple of two, each step of the lazy random walk can be implemented with a number chosen uniformly at random from $\{1,2, \ldots, 2d\}$, which can be implemented with $1+\log d = O(1)$ fair coin flips. 
We will implement these lazy random walks with $k$-wise independent fair coin flips, with 
\[k =  (1+\log d) \cdot 2r \cdot \tmix(\widetilde{\Gsp}).\]
Therefore, for any choice of at most $2r$ lazy random walks, they behave the same as lazy random walks implemented with mutually independent random variables.

More formally, these $k$-wise independent random variables are implemented as a hash function $h$ mapping from $(\alpha, \beta, \gamma)$ to a number in $\{1,2, \ldots, 2d\}$ such that $h(\alpha, \beta, \gamma)$ encodes the decision of the $\alpha$th step of the $\beta$th lazy random walk associated with a message sent from a vertex $v \in V$ with $\ID(v) = \gamma$. Since the number of possible values of $\alpha$, $\beta$, and $\gamma$ are upper bounded by $\poly(n)$, these $k$-wise independent random variables can be implemented with
$O(k \log n)$ mutually independent coin flips.

\paragraph{Good messages.} Recall that each message is associated with $r$ lazy random walks starting from the same vertex in $\Vsp$. Given the hash function $h$, we say that a message is \emph{good} if the following two conditions are satisfied.
\begin{itemize}
    \item At least one of these $r$ lazy random walks ends at a vertex in $X_{\vstar}$.
    \item For each $w \in \Vsp$ and each integer $t$ such that at least one of these $r$ lazy random walks is at $w$ at time step $t$, the total number of lazy random walks that are at $w$ at time step $t$ is at most $3r$.
\end{itemize}

All the good messages can be delivered to $\vstar$ in $3r \cdot \tmix(\widetilde{\Gsp})$ rounds by simply simulating all the lazy random walks for $\tmix(\widetilde{\Gsp})$ steps, where we allocate $3r$ rounds for each step. This number of rounds is sufficient to simulate one step as long as the total number of lazy random walks at a vertex $u \in \Vsp$ at time step $t$  is at most $3r$, for each  $u \in \Vsp$ and for each time step $t$. If this bound does not hold for some $u$ and $t$, then we simply discard all random walks at a vertex $u \in \Vsp$ at time step $t$.  As long as a message is good, it is guaranteed that at least one walk associated with the message successfully ends at $X_{\vstar}$, so the message is successfully delivered. 

\paragraph{Analysis.}
We will show that each message is good with probability at least $1-f$, so the above algorithm delivers at least $1-f$ fraction of the messages in expectation.

\begin{lemma}\label{lem:goodmsg-1}
For each $u \in \Vsp$, with probability at least $1 - \frac{f}{2}$, at least one of the $r$ lazy random walks starting from $u$ ends at a vertex in $X_{\vstar}$.
\end{lemma}
\begin{proof}
As discussed earlier, our choice of the value $k$ implies that these $r$ lazy random walks starting from $u$ behave the same as lazy random walks implemented with mutually independent random variables. By the definition of $\tmix$, for each of these random walks, the probability that it ends at a vertex in $X_{\vstar}$ is at least \[\frac{|X_{\vstar}|}{|\Vsp|} \cdot \left(1 - \frac{1}{|\Vsp|}\right) \geq \frac{|X_{\vstar}|}{2|\Vsp|}.\] 
Therefore, the probability that none of the $r$ lazy random walks starting from $u$ ends at a vertex in $X_{\vstar}$ is at most
\[\left(1 - \frac{|X_{\vstar}|}{2|\Vsp|}\right)^r < \left(1 - \frac{|X_{\vstar}|}{2|\Vsp|}\right)^{C \cdot \frac{|\Vsp|}{|X_{\vstar}|} \cdot \log \frac{1}{f}} \leq e^{-\frac{C}{2} \cdot \log \frac{1}{f}}.\]
By selecting $C$ to be a sufficiently large number, the above probability can be made at most $\frac{f}{2}$.
\end{proof}

In the following lemma, we write $L_{w,t}$ to denote the load of vertex $w \in \Vsp$ at time step $t$, which is defined as the number of lazy random walks at $w$ at time step $t$.

\begin{lemma}\label{lem:goodmsg-2}
For each $u \in \Vsp$, with probability at least $1 - \frac{f}{2}$, we have $L_{w,t} \leq 3r$ for each $w \in \Vsp$ and each time step $t$ such that at least one of the $r$ lazy random walks starting from $u$ is at $w$ at time step $t$.
\end{lemma}
\begin{proof}
As each vertex $w \in \Vsp$ initially has exactly the same load $L_{w,0} = r$, we  have
\[\Expect[L_{w,t}] = r\]
for all $w \in \Vsp$ and all $t$. This observation follows immediately from the formula for the probability distribution for the next step of the lazy random walk: $\frac{1}{2} \cdot p + \frac{1}{2} \cdot ADp$, which equals $p$ when $p$ is a uniform distribution and the underlying graph is regular.

 We first reveal and fix all the $r$ lazy random walks starting from $u$, and we will do the analysis using only the remaining unrevealed randomness. We define $L_{w,t}'$ as the load of $w$ at time step $t$ excluding the contribution of the random walks starting from $u$, so we always have \[L_{w,t} \leq L_{w,t}' + r.\] 
 Since the number of pairs $(w,t)$ considered by this lemma  is at most $r \tmix(\widetilde{\Gsp})$, to prove this lemma it suffices to show that 
 \[\Pr[L_{w,t}' \geq 2r] \leq \frac{f}{2r \tmix(\widetilde{\Gsp})}\]
for all $w \in \Vsp$ and time steps $t$ such that at least one of the $r$ lazy random walks starting from $u$ is at $w$ at time step $t$. If the above probability bound holds, then by a union bound, with probability at least $1-\frac{f}{2}$, for all pairs $(w,t)$ considered by this lemma, we have $L_{w,t} \leq L_{w,t}' + r \leq 3r$, as required.

As discussed earlier, our choice of the value $k$ implies that any choices of at most $2r$ lazy random walks behave the same as lazy random walks implemented with mutually independent random variables. Therefore, even after revealing all the $r$ lazy random walks starting from $u$, any choice of at most $r$ remaining lazy random walks still behave the same as lazy random walks implemented with mutually independent random variables. We may write $L_{w,t}'$ as a summation $X = \sum_{i=1}^s x_i$, where $x_i \in \{0,1\}$ is the indicator random variable for the event that the $i$th remaining lazy random walk is at $w$ at time step $t$, and $s$ is the total number of remaining lazy random walks. The above discussion implies that the random variables $x_1, x_2, \ldots, x_s$ are $r$-wise independent, so we may apply the tail bound of~\cite{SchmidtSS95} mentioned earlier.
Let $\mu = \Expect[X]$ and select $\delta$ in such a way that $\delta \mu = r$. Since we know that $\mu \leq r$, we have $\delta \geq 1$, so the tail bound implies that 
\[\Pr\left[L_{w,t}' \geq 2r\right] \leq \Pr[X \geq (1+\delta)\mu] \leq e^{-\frac{\mu \delta}{3}} = e^{-\frac{r}{3}}.\]
Since $r = C \cdot \left( \frac{|\Vsp|}{|X_{\vstar}|} \cdot \log \frac{1}{f} + \log \tmix(\widetilde{\Gsp})\right)$, by selecting $C$ to be a sufficiently large number, the above probability can be made at most $\frac{f}{2r \tmix(\widetilde{\Gsp})}$, as required.
\end{proof}

\paragraph{Derandomization.}
\cref{lem:goodmsg-1,lem:goodmsg-2} imply that each message is good with probability at least $1-f$, so the algorithm described above solves the information-gathering problem in expectation. We can derandomize the algorithm without worsening the asymptotic round complexity.

\begin{lemma}\label{lem:gathering-2}
    Let $G=(V,E)$ be a $\phi$-expander, and select $\vstar \in V$  such that $\deg(\vstar)$ equals the maximum degree $\Delta$ of $G$. Suppose each vertex $v \in V$ wants to 
send $\deg(v)$ messages of $O(\log n)$ bits  to $\vstar$.
For any $0 < f < \frac{1}{2}$, there exists a routing schedule that can be encoded as a string of 
\[O\left(\frac{|E|}{\Delta} \cdot  \log \frac{1}{f} + \log \phi^{-1} + \log \log |E|\right) \cdot O(\phi^{-2} \log |E|) \cdot O(\log n)\]
bits such that if all vertices in $V$ know the string, then
 there is an algorithm that takes \[O\left(\frac{|E|}{\Delta} \cdot  \log \frac{1}{f} + \log \phi^{-1} + \log \log |E|\right) \cdot O(\phi^{-2} \log |E|) \] rounds and delivers at least $1-f$ fraction of these $2|E|$ messages to $\vstar$.
\end{lemma}
\begin{proof}
 By \cref{lem:goodmsg-1,lem:goodmsg-2}, each message is good with probability at least $1-f$, so there exists a choice of a hash function $h$ to make the algorithm described above deliver at least $1-f$ fraction of the messages. As discussed earlier, such a function $h$ can be described with a string of
 \[O(k \log n) = O\left(\frac{|E|}{\Delta} \cdot  \log \frac{1}{f} + \log \phi^{-1} + \log \log |E|\right) \cdot O(\phi^{-2} \log |E|) \cdot O(\log n)\] bits.
 Once all vertices in $V$ know the string, then the routing algorithm costs 
 \[3r \cdot \tmix(\widetilde{\Gsp}) =  O\left(\frac{|E|}{\Delta} \cdot  \log \frac{1}{f} + \log \phi^{-1} + \log \log |E|\right) \cdot O(\phi^{-2} \log |E|) \] rounds of communication.
\end{proof}

In the $\CONGEST$ model, in each round, $O(\log n)$ bits can be transmitted along each edge, so the cost of disseminating the routing schedule is linear in the round complexity of the routing algorithm.
Comparing with \cref{lem:gathering-1}, there is an additional term $\log \phi^{-1} + \log \log |E|$ in the round complexity of \cref{lem:gathering-1} which is due to $\log \tmix(\widetilde{\Gsp})$.
When $\log \phi^{-1} + \log \log |E| = O\left(\frac{|E|}{\Delta} \cdot  \log \frac{1}{f}\right)$, \cref{lem:gathering-2} is more efficient than \cref{lem:gathering-1} by a factor of $O\left(\log \frac{1}{f}\right)$. 

\paragraph{Multiple subgraphs.} Consider the setting where we have multiple subgraphs $G_1, G_2, \ldots, G_s$ of $G=(V,E)$ and there exists a vertex $v'$ that knows the graph topology of all these subgraphs. In the following lemma, we extend \cref{lem:gathering-2} to this setting to show that $v'$ can prepare one succinct routing schedule that can be used by all these subgraphs to solve the information-gathering problem, where the guarantee is that at least $1-f$ fraction of the messages, among all messages in  $G_1, G_2, \ldots, G_s$, are successfully delivered.

\begin{lemma}\label{lem:gathering-3}
    Let $G=(V,E)$ be the underlying communication network.
    Let $G_1=(V_1, E_1), G_2=(V_2, E_2), \ldots, G_s=(V_s, E_s)$ be $s$ disjoint subgraphs of $G$.
    Let $\vstar_i$ be a vertex in $G_i$ such that its degree in $G_i$ equals the maximum degree $\Delta_i$ of $G_i$. We write $\eta$ to denote the maximum value of $\frac{|E_i|}{|\Delta_i|}$ over all $1 \leq i \leq s$, and we write $\zeta$ to denote the maximum value of $|E_i|$ over all $1 \leq i \leq s$.
    
    Suppose each $G_i$ is a $\phi$-expander and each vertex $v \in V_i$ wants to 
send $\deg_{G_i}(v)$ messages of $O(\log n)$ bits  to $\vstar_i$.
For any $0 < f < \frac{1}{2}$, there exists a routing schedule that can be encoded as a string of 
\[O\left(\eta \cdot  \log \frac{1}{f} + \log \phi^{-1} + \log \log \zeta\right) \cdot O(\phi^{-2} \log \zeta) \cdot O(\log n)\]
bits such that if all vertices in $G_1, G_2, \ldots, G_s$ know the string, then
 there is an algorithm that takes \[O\left(\eta \cdot  \log \frac{1}{f} + \log \phi^{-1} + \log \log \zeta\right) \cdot O(\phi^{-2} \log \zeta) \] rounds and delivers at least $1-f$ fraction of these $2 \sum_{i=1}^s |E_i|$ messages.  The routing schedule can be computed given the graph topology of $G_1, G_2, \ldots, G_s$ and the selected vertices $\vstar_1, \vstar_2, \ldots, \vstar_s$.
\end{lemma}
\begin{proof}
The proof of this lemma follows from the proof of \cref{lem:gathering-2} with a minor modification. 
Here we use the same hash function $h$ for all of  $G_1, G_2, \ldots, G_s$, where the values of $k$ and $r$ are chosen to be the maximum over all $G_1, G_2, \ldots, G_s$. We use $O(\phi^{-2} \log \zeta)$ for the mixing time upper bound for all these subgraphs.
 By \cref{lem:goodmsg-1,lem:goodmsg-2}, each message is good with probability at least $1-f$, so there exists a choice of a hash function $h$ to make the routing algorithm deliver at least $1-f$ fraction of the messages, among all messages in $G_1, G_2, \ldots, G_s$. Therefore, we obtain the same string length and the same round complexity as the ones in \cref{lem:gathering-2}, except that $\frac{|E|}{\Delta}$ is replaced with $\eta$, and $|E|$ is replaced with $\zeta$.
\end{proof}

\subsection{Minor-free networks}\label{sect:routing-minor-free}

It is known that any $H$-minor-free $\phi$-expander must have a very large maximum degree~\cite{10.1145/3519270.3538423}. 

\begin{lemma}[\cite{10.1145/3519270.3538423}]\label{lem:separator}
If $G=(V,E)$ is an $H$-minor-free $\phi$-expander with  maximum degree $\Delta$, then  $\Delta = \Omega(\phi^2) |V|$. 
\end{lemma}

It is well-known that $H$-minor-free graphs have bounded arboricity, so $|E| = O(|V|)$.
Therefore, for the case of $H$-minor-free graphs, \cref{lem:separator} implies that the factor  $O\left(\frac{|E|}{\Delta}\right)$ in the round complexities of our information-gathering algorithms can be replaced with $O(\phi^{-2})$. Moreover, since $|E| = O(|V|) = O(\Delta \phi^{-2})$, the factor $O(\log |E|)$ can be replaced with $O(\log \Delta + \log \phi^{-1})$.

For the case where $G$ is $H$-minor-free, the round complexity of \cref{lem:gathering-1} becomes \[O\left(\phi^{-4} \log |E| \log^2 \frac{1}{f}\right).\] If we want to deliver all messages, then we may set $f = \frac{1}{2|E|+1}$, in which case the round complexity becomes $O(\phi^{-4} \log^3 |E|)$. 

For \cref{lem:gathering-2}, if $G$ is $H$-minor-free, then the length of the routing schedule becomes 
\[O\left(\phi^{-2} \log \frac{1}{f} + \log \log |E|\right)\cdot O\left(\phi^{-2} \log |E|\right) \cdot O(\log n)\] bits and the round complexity of the routing algorithm becomes  \[O\left(\phi^{-2} \log \frac{1}{f} + \log \log |E|\right)\cdot O\left(\phi^{-2} \log |E|\right).\] If we want to deliver all messages, then we may set $f = \frac{1}{2|E|+1}$, in which case the routing schedule can be encoded in $O(\phi^{-4} \log^2 |E|) \cdot O(\log n)$ bits and the round complexity is $O(\phi^{-4} \log^2 |E|)$, which is more efficient than \cref{lem:gathering-1} by a factor of $O(\log |E|)$.

\section{Simple existential bounds}\label{sect:existential}

An $(\epsilon, \phi)$-\emph{expander decomposition} of a graph $G=(V,E)$ is a partition of its vertex set into clusters $V = V_1 \cup V_2 \cup \cdots \cup V_k$   meeting the following two conditions.
\begin{itemize}
    \item The number of inter-cluster edges $\frac{1}{2}\sum_{i=1}^k |\partial(V_i)|$ is at most   $\epsilon|E|$.
    \item For each $1 \leq i \leq k$, either $|V_i|=1$ or the subgraph $G[V_i]$ induced by $V_i$ is a $\phi$-expander.
\end{itemize}

We first show that any $H$-minor-free graph with maximum degree $\Delta$ admits an $(\epsilon, \phi)$ expander decomposition with $\phi = \Omega\left(\frac{\epsilon}{\log \frac{1}{\epsilon} + \log \Delta}\right)$, independent on the number of vertices $|V|$. 
We begin with the following well-known fact, for which we include a proof for the sake of completeness.

\begin{fact}\label{fact:existence-basic}
For any $\epsilon \in (0,1)$, any graph $G=(V,E)$ has an $\left(\epsilon,\Omega\left(\frac{\epsilon}{\log |V|}\right)\right)$ expander decomposition.
\end{fact}
\begin{proof}
Let $\phi = \frac{\epsilon}{4\log |V|}$.
    Suppose the graph $G=(V,E)$ has a cut $(S, V \setminus S)$ with $\Phi(S) \geq \phi$. We delete all the edges crossing $S$ and $V \setminus S$ and recurse on the two induced subgraphs $G[S]$ and $G[V \setminus S]$. In the end, the subgraph induced by each remaining connected component of size at least two is guaranteed to have conductance at least $\phi$. We claim that this is an $(\epsilon, \phi)$ expander decomposition. To prove this claim, we just need to show that the number of deleted edges is at most $\epsilon|E|$.  
    Whenever a cut $(A,B)$ is found during the algorithm, we consider the following charging scheme. 
    Without loss of generality, we assume $\vol(A) \geq \vol(B)$.
    We charge the cost of the deleted edges uniformly to each pair $(v,e)$ where $v \in B$ and $e$ is an edge incident to $v$. As the conductance of the cut is at most $\phi$, the cost charged to $(v,e)$ is  $\frac{|E(A,B)|}{\vol(B)} \leq \phi$.
    Since $\vol(B) \leq \frac{\vol(A\cup B)}{2}$, the charging scheme guarantees that each pair $(v,e)$ is charged at most $\log \vol(V) < 2 \log |V|$ times, so the total cost is upper bounded by $\phi \cdot \vol(V) \cdot 2 \log |V| < \epsilon |E|$, as required. 
\end{proof}

It is well-known~\cite{klein1993excluded,Fakcharoenphol2003improved,Ittai2019padded} that for any $H$-minor-free graph, an $(\epsilon, D)$  low-diameter decomposition  with $D = O(\epsilon^{-1})$ exists.

\begin{lemma}[\cite{klein1993excluded,Fakcharoenphol2003improved,Ittai2019padded}]\label{lem:LDD-basic}
For any $\epsilon \in (0,1)$, any $H$-minor-free graph $G=(V,E)$ has an $\left(\epsilon, O\left(\frac{1}{\epsilon}\right)\right)$ low-diameter decomposition.
\end{lemma}

Using \cref{lem:separator,lem:LDD-basic}, the conductance bound $\phi = \Omega\left(\frac{\epsilon}{  \log |V|}\right)$ of \cref{fact:existence-basic} can be improved to  $\phi = \Omega\left(\frac{\epsilon}{\log \frac{1}{\epsilon} + \log \Delta}\right)$  for $H$-minor-free graphs $G=(V,E)$.

\begin{observation}\label{obs:existence}
For any $\epsilon \in (0,1)$, any $H$-minor-free graph $G=(V,E)$ with maximum degree $\Delta$ has an $\left(\epsilon, \phi \right)$ expander decomposition with $\phi = \Omega\left(\frac{\epsilon}{\log \frac{1}{\epsilon} + \log \Delta}\right)$.
\end{observation}
\begin{proof}
  The expander decomposition is constructed in three steps. In the first step, we compute an $\left(\frac{\epsilon}{3}, O\left(\frac{1}{\epsilon}\right)\right)$ low-diameter decomposition using \cref{lem:LDD-basic}. Each cluster of the low-diameter decomposition has diameter $D = O\left(\frac{1}{\epsilon}\right)$, so the cluster contains at most $n' = \Delta^{O\left(\frac{1}{\epsilon}\right)}$ vertices. 
  
  In the second step, we further refine the current clustering by computing an $\left(\frac{\epsilon}{3},\Omega\left(\frac{\epsilon}{\log n'}\right)\right)$ expander decomposition of each cluster using \cref{fact:existence-basic}. After the decomposition, the subgraph induced by each cluster has conductance at least $\phi' = \Omega\left(\frac{\epsilon}{\log n'}\right) = \Omega\left(\frac{\epsilon^2}{\log \Delta}\right)$, so \cref{lem:separator} implies that the cluster has at most $n'' = \Delta \cdot O\left(\frac{1}{{\phi'}^2}\right) = O\left(\frac{\Delta \log^2 \Delta}{\epsilon^4}\right)$ vertices.

  In the third step, we refine the current clustering again by computing an $\left(\frac{\epsilon}{3},\Omega\left(\frac{\epsilon}{\log n''}\right)\right)$ expander decomposition of each cluster using \cref{fact:existence-basic}. After that, the conductance of each cluster of the decomposition is at least $\phi'' = \Omega\left(\frac{\epsilon}{\log n''}\right) = \Omega\left(\frac{\epsilon}{\log \frac{1}{\epsilon} + \log \Delta}\right)$, meeting the conductance requirement in the lemma statement. The decomposition  is a desired  expander decomposition, as the total number of inter-cluster edges is at most $3 \cdot \frac{\epsilon}{3} = \epsilon$ fraction of the edge set $E$.
\end{proof}

\section{Expander decompositions with overlaps}\label{sect:existence-variant}

In this section, we consider a variant of expander decomposition that allows the $\phi$-expanders to slightly overlap, and we will show that all $H$-minor-free graphs $G=(V,E)$ admit this variant of expander decompositions with $\phi = 2^{-O\left(\log^2 \frac{1}{\epsilon}\right)}$, which is independent of the number of vertices $|V|$ and the maximum degree $\Delta$ of the graph.

Formally, we define an $(\epsilon,\phi,c)$ \emph{expander decomposition} of a graph $G=(V,E)$ as a partition of its vertex set into clusters $V = V_1 \cup V_2 \cup \cdots \cup V_k$, where each cluster $S \in \VV = \{V_1, V_2, \ldots, V_k\}$ is associated with a subgraph $G_{S}$ of $G$
such that  $G[S]$ is a subgraph of $G_{S}$, meeting the following three conditions.
\begin{itemize}
    \item The number of inter-cluster edges $\frac{1}{2}\sum_{i=1}^k |\partial(V_i)|$ is at most   $\epsilon|E|$.
    \item For each cluster $S \in \VV$, either $|V(G_S)| = 1$ or  $G_S$ is a $\phi$-expander.
    \item For each vertex $v \in V$, we have 
    $|\{ S \in \VV \mid v \in V(G_S)\}|\leq c$.
\end{itemize}

Observe that $c \geq 1$, and that setting $c = 1$ forces $G_S = G[S]$ for every cluster $S \in \VV$, since by definition $G[S]$ must be a subgraph of $G_S$. Therefore, an $(\epsilon,\phi,1)$ expander decomposition coincides with a standard $(\epsilon,\phi)$ expander decomposition.
The main goal of this section is to prove the following lemma.

\begin{lemma}\label{lem:existence-variant}
For any $\epsilon \in \left(0,\frac{1}{2}\right)$, any $H$-minor-free graph $G=(V,E)$  admits an $\left(\epsilon, \phi, c\right)$ expander decomposition with $\phi = 2^{-O\left(\log^2 \frac{1}{\epsilon}\right)}$ and $c = O\left(\log \frac{1}{\epsilon}\right)$.
\end{lemma}

The above lemma is proved by modifying the low-diameter decomposition algorithm of Czygrinow, Ha{\'n}{\'c}kowiak, and Wawrzyniak~\cite{czygrinow2008fast}. We begin by reviewing their algorithm in \cref{sect:ldd-review}, and then present our modification in \cref{sect:ldd-modify}. As we will show later, the decomposition of \cref{lem:existence-variant} can in fact be computed deterministically in the $\CONGEST$ model within 
$2^{O\left(\log^2 \frac{1}{\epsilon}\right)} \cdot  O(\log \Delta)  \cdot \left(O(\log^3 \Delta) +O(\log^\ast n)\right)$ rounds, using the routing algorithm from \cref{lem:gathering-1}.

\subsection{The heavy-stars algorithm}\label{sect:ldd-review}

The \emph{arboricity} of a graph $G=(V,E)$ is the smallest number $\alpha$ such that the edge set $E$ can be partitioned into $\alpha$ edge-disjoint forests. It is well-known that any $H$-minor-free graph has arboricity $\alpha = O(1)$, for any fixed $H$. Specifically, it was shown in~\cite{THOMASON2001318} that $\alpha = O(t \sqrt{\log t})$, where $t$ is the number of vertices in $H$. Throughout the paper, we let $\alpha = O(1)$ be an arboricity upper bound of any $H$-minor-free graph.

\paragraph{The heavy-stars algorithm.}
Suppose each edge $e \in E$ in a $H$-minor-free graph $G=(V,E)$ has a positive integer weight $w(e)$. The \emph{heavy-stars} algorithm of Czygrinow, Ha{\'n}{\'c}kowiak, and Wawrzyniak~\cite{czygrinow2008fast} finds a set of vertex-disjoint stars of $G$ containing at least $\frac{1}{8 \alpha}$ fraction of the edge weights in $O(\log^\ast n)$ rounds deterministically in the $\LOCAL$ model.
In the subsequent discussion, we slightly abuse the notation to write $w(S) = \sum_{e \in S} w(e)$ for any edge set $S$.
 We now describe the heavy-stars algorithm of~\cite{czygrinow2008fast}. 
\begin{description}
    \item[Step 1: edge orientation.]  Each vertex $u \in V$ picks an edge $e$ incident to $u$ that has the highest weight among all edges  incident to $u$, breaking the tie by selecting the edge $e=\{u,v\}$ such that $\ID(u)+\ID(v)$ is maximized. The edge $e=\{u,v\}$ selected by $u$ is oriented as  $u \rightarrow v$. If $e=\{u,v\}$ is selected by both $u$ and $v$, then  $e=\{u,v\}$ is  oriented as either  $u \rightarrow v$ or  $u \leftarrow v$ arbitrarily. The tie-breaking mechanism ensures that the oriented edges do not form cycles.
    \item[Step 2: vertex coloring.]  The oriented edges induce a set of vertex-disjoint rooted trees $\{T_i\}$. For each rooted tree $T_i$, compute a proper $3$-coloring. This can be done in $O(\log^\ast n)$ rounds deterministically using the Cole--Vishkin algorithm~\cite{ColeV86}.
    \item[Step 3: low-diameter clustering.] For each vertex $u \in V$ and each color subset $C\subseteq \{1,2,3\}$, we write $\inn(u,C)$ to denote the set of outgoing edges $u \rightarrow v$ incident to $u$ such that the color of $v$ is in $C$, and we write  $\outt(u,C)$ to denote the set of incoming edges $u \leftarrow v$ incident to $u$ such that the color of $v$ is in $C$.
    \begin{enumerate}
        \item  For each $u \in V$ that is colored $1$, if $w(\inn(u,\{2,3\})) \geq w(\outt(u,\{2,3\}))$, then $v$ marks all the edges in $\inn(u,\{2,3\})$, otherwise $v$ marks the unique edge in $\outt(u,\{2,3\})$.
        \item For each $u \in V$ that is colored $2$, if $w(\inn(u,\{3\})) \geq w(\outt(u,\{3\}))$, then $v$ marks all the edges in $\inn(u,\{3\})$, otherwise $v$ marks the unique edge in $\outt(u,\{3\})$. 
    \end{enumerate}
    \item[Step 4: star formation.] The marked edges induce vertex-disjoint rooted trees $\{Q_i\}$ of depth at most $4$. For each rooted tree $Q_i$, find a set of vertex-disjoint stars in $Q_i$ covering at least half of the edge weights in $Q_i$. Such a set of vertex-disjoint stars can be found by either taking all the edges from odd levels to even levels or taking all the edges from even levels to odd levels. 
\end{description}

\paragraph{Analysis.} It was shown in~\cite{czygrinow2008fast} that the vertex-disjoint stars constructed by the algorithm capture at least $\frac{1}{8 \alpha}$ fraction of the edge weights and that the algorithm indeed can be implemented in $O(\log^\ast n)$ rounds in the $\LOCAL$ model. We include a proof here for the sake of completeness. 

\begin{lemma}[\cite{czygrinow2008fast}]\label{lem:heavystar-weight}
The vertex-disjoint stars constructed by the heavy-stars algorithm capture at least $\frac{1}{8 \alpha}$ fraction of the edge weights.
\end{lemma}
\begin{proof}
Observe that the vertex-disjoint rooted trees $\{T_i\}$ capture at least $\frac{1}{2 \alpha}$ fraction of the edge weights. To see this, we partition the edge set into $\alpha$ forests and pick a forest $F$ that has the highest weight, so $w(F) \geq \frac{|E|}{ \alpha}$. We associate each edge $e$ in the forest to a distinct vertex $v \in V$. By the description of the heavy-tree algorithm, the weight of the edge picked by $u \in V$ is at least the weight of the edge in $F$ assigned to $u$. Since each oriented edge $e$ is picked by at most two vertices, the set of oriented edges has a total weight of at least $\frac{w(F)}{2} \geq \frac{|E|}{2 \alpha}$.

To show that the vertex-disjoint stars constructed by the algorithm capture at least $\frac{1}{8 \alpha}$ fraction of the edge weights, we just need to show that  $\{Q_i\}$ capture at least $\frac{1}{4 \alpha}$ fraction of the edge weights. This follows from the observation that $\{Q_i\}$ capture at least half of the edge weights in $\{T_i\}$. This observation follows from the fact that the edge sets considered in Step 3.1, $\inn(u,\{2,3\})$ and  $\outt(u,\{2,3\})$, ranging over all $u \in V$  colored $1$, together with the edge sets considered in Step 3.2, $\inn(u,\{3\})$ and  $\outt(u,\{3\})$, ranging over all $u \in V$  colored $2$, form a partition of the set of all edges in $\{T_i\}$.
\end{proof}

It is clear that Steps 1 and 3 take only $O(1)$ rounds.
To show that the algorithm can be implemented to run in $O(\log^\ast n)$ rounds in the $\LOCAL$ model, we just need to prove that the marked edges induce vertex-disjoint rooted trees $\{Q_i\}$ of depth at most $4$, so the construction of the stars in Step 4 can also be done in $O(1)$ rounds.

\begin{lemma}[\cite{czygrinow2008fast}]\label{lem:heavystar-depth}
In the heavy-stars algorithm, the marked edges induce vertex-disjoint rooted trees $\{Q_i\}$ of depth at most $4$
\end{lemma}
\begin{proof}
Suppose that there is a directed path $P = v_0 \rightarrow v_1 \rightarrow v_2 \rightarrow v_3 \rightarrow v_4 \rightarrow v_5$ in $Q_i$ of length $5$. By the description of Step 3.1, $v_1$, $v_2$, $v_3$, and $v_4$ cannot be colored $1$. Therefore, there exists $j \in \{2,3\}$ such that $v_j$ is colored $2$ and both $v_{j-1}$ and $v_{j+1}$ are  colored $3$, which is impossible by the  description of Step 3.2, so the rooted tree $Q_i$ has depth at most $4$. 
\end{proof}

\paragraph{Low-diameter decomposition via the heavy-stars algorithm.}
The heavy-stars algorithm was used to obtain a low-diameter decomposition with $D = \poly\left(\frac{1}{\epsilon}\right)$ for $H$-minor-free graphs in~\cite{czygrinow2008fast}. For any given partition of the vertex set $V$, we define the \emph{cluster graph} as follows. Each cluster is a vertex in the cluster graph. Two clusters are adjacent if there exists an edge crossing the two clusters. Set the weight of an edge between two clusters as the number of edges crossing the two clusters.

The low-diameter decomposition is as follows. Start with a trivial clustering where each vertex $v \in V$ is a cluster. Repeat the following procedure for $O\left(\log \frac{1}{\epsilon}\right)$ iterations: Apply  the heavy-stars algorithm to the cluster graph and merge each star into a cluster. 
In each iteration, the cluster diameter is increased by at most constant factor, so in the end, the diameter of each cluster is at most $D = \poly\left(\frac{1}{\epsilon}\right)$. In each iteration, the number of inter-cluster edges is reduced by a factor of at most $1 - \frac{1}{8 \alpha}$, so $O\left(\log \frac{1}{\epsilon}\right)$ iterations suffice to make the number of 
inter-cluster edges to be at most $\epsilon \cdot |E|$. The algorithm can be implemented to run in  $\poly\left(\frac{1}{\epsilon}\right) \cdot O(\log^\ast n)$ rounds in the $\LOCAL$ model~\cite{czygrinow2008fast}, as each round in the cluster graph can be simulated by $O(D) = \poly\left(\frac{1}{\epsilon}\right)$ rounds in the original graph $G$.

\paragraph{The message-size complexity.}
A natural question to ask is whether the   above low-diameter decomposition algorithm admits an efficient implementation in the $\CONGEST$ model.
In the low-diameter decomposition algorithm, the heavy-stars algorithm is executed in the cluster graph where the diameter of each cluster is at most $D = \poly\left(\frac{1}{\epsilon}\right)$.
We examine the message-size complexity of each step of the heavy-stars algorithm in a cluster graph. Steps~3~and~4 can be implemented in $O(D)$ rounds in the $\CONGEST$ model because these steps only involve calculating a summation of some edge weights a constant number of times.

Step~2 can be implemented to run in $O(D) \cdot O(\log^\ast n)$ rounds in the $\CONGEST$ model because in each round of the Cole--Vishkin algorithm~\cite{ColeV86}, each cluster $u$ can update its color based on the current color of $u$ and the current color of the parent of $u$. Therefore, to implement one round of the Cole--Vishkin algorithm, each cluster $u$ just needs to receive an $O(\log n)$-bit message from its parent cluster $v$, where the message encodes the color of $v$,  and to send an $O(\log n)$-bit to all children of $u$, where the message encodes the color of $u$.

Step 1 is the only part of the heavy-stars algorithm that does not appear to admit an efficient implementation in the $\CONGEST$ model, as it requires computing, for each neighboring cluster, the number of incident edges, and then identifying the maximum. This bottleneck is precisely why the above low-diameter decomposition is not efficient in the $\CONGEST$ model. To address this issue, Levi, Medina, and Ron~\cite{levi2021property} proposed two approaches for obtaining an efficient low-diameter decomposition in $\CONGEST$ by modifying Step 1. The first approach incurs logarithmic rounds to partition the edge set into $O(1)$ forests, while the second relies on randomness. Neither approach is suitable for our purposes, as we aim for a deterministic algorithm with sublogarithmic round complexity.

\paragraph{Overcoming the issue.}
We overcome this issue by modifying the heavy-stars algorithm. The key idea is to ensure that each cluster has high conductance, enabling the application of the information-gathering tool of \cref{lem:gathering-1}. First, we remove vertices that are too weakly connected to their current cluster, so that in every remaining non-singleton cluster, each vertex has internal degree comparable to its total degree; this keeps the intra-cluster communication load under control. Next, we form heavy stars in the cluster graph, but before merging them, we discard light links to guarantee that every resulting cluster maintains sufficiently large conductance.

Together, these modifications create the conditions needed to apply the information-gathering tool of \cref{lem:gathering-1} efficiently within each cluster. Consequently, the relevant edge-count data can be routed to a designated vertex, which can locally determine the heaviest link and thereby simulate Step 1 in the $\CONGEST$ model. Moreover, since the modified algorithm ensures high conductance for each cluster, it enables us to prove \cref{lem:existence-variant}.

\subsection{The expander decomposition algorithm}\label{sect:ldd-modify}

In this section, we prove \cref{lem:existence-variant} by modifying the low-diameter decomposition algorithm of~\cite{czygrinow2008fast}. Specifically, we  use the heavy-stars algorithm to prove the following lemma.

\begin{lemma}\label{lem:existence-variant-aux}
Let $\epsilon \in [0,1]$, $\phi \in [0,1]$, and $c \geq 1$.
Let $G=(V,E)$ be an $H$-minor-free graph.
If $G$ admits an $\left(\epsilon, \phi, c\right)$ expander decomposition, then $G$ also admits an $\left(\epsilon', \phi', c'\right)$ expander decomposition with $\epsilon' = \epsilon \cdot \left(1-\frac{1}{32\alpha}\right)$, $\phi' = \phi \cdot \frac{\epsilon}{13056\alpha^2 c(c+1)}$, and $c' = c+1$.
\end{lemma}

We first prove \cref{lem:existence-variant} using \cref{lem:existence-variant-aux}.

\begin{proof}[Proof of \cref{lem:existence-variant}]
Start with the trivial $\left(1, 1, 1\right)$ expander decomposition where each vertex is a cluster, and the edge set $E_i$ associated with each cluster is an empty set. Applying \cref{lem:existence-variant-aux} for 
\[t = \left\lceil \frac{\log \frac{1}{\epsilon}}{\log \frac{1}{1-\frac{1}{32\alpha}}}\right\rceil = O\left(\log \frac{1}{\epsilon}\right)\]
iterations suffices to reduce the number of inter-cluster edges to at most $\epsilon |E|$, and this gives us an $\left(\epsilon, \phi, c\right)$ expander decomposition with 
\[\phi \geq \left(\frac{\epsilon}{13056\alpha^2 t^2}\right)^t = \epsilon^{O(t)} =  2^{-O\left(\log^2 \frac{1}{\epsilon}\right)} \ \ \ \text{and} \ \ \ 
    c = t+1 = O\left(\log \frac{1}{\epsilon}\right),\]
as required.
\end{proof}

We need the following auxiliary lemma to prove \cref{lem:existence-variant-aux}.

\begin{lemma}\label{lem:merge-expander-aux}
Let $\epsilon \in [0,1]$, $\phi \in [0,1]$, and $d \geq 1$.
Let $G=(V,E)$ be any graph.
Let $V = V_1 \cup V_2 \cup \cdots \cup V_k$ be a partition satisfying the following conditions:
\begin{itemize}
    \item For each $i \in [k]$, either $|V_i|=1$ or $G[V_i]$ is a $\phi$-expander.
    \item For each $i \in [k]$  such that $|V_i| > 1$, and for each $v \in V_i$, we have $|E(\{v\}, V_i)|\geq \frac{1}{d} \cdot \deg(v)$.
    \item For each $i \in [2, k]$, $|E(V_1, V_i)| \geq \epsilon \cdot \vol(V_i)$.
\end{itemize}
Then $G$ is a $\phi'$-expander with $\phi' = \frac{\epsilon \phi}{6d}$.
\end{lemma}
\begin{proof}
We begin with the following observation. Consider any cut $(S, V \setminus S)$ such that for each $i \in [k]$, either $V_i \subseteq S$ or $V_i \subseteq V \setminus S$. We claim that $\Phi(S) \geq \epsilon$. To prove this claim, without loss of generality we may assume $V_1 \subseteq V \setminus S$, so
\[
|E(S, V \setminus S)| \geq \sum_{i \in [2,k] \; : \; V_i \subseteq S} |E(V_1, V_i)| \geq \epsilon \cdot  \sum_{i \in [2,k] \; : \; V_i \subseteq S} \vol(V_i) = \epsilon \vol(S) \geq \epsilon \min\{\vol(S), \vol(V \setminus S)\},
\]
which implies $\Phi(S) \geq \epsilon$.

The rest of the proof is similar to \cite[Lemma C.1]{ChangS20}. We divide the analysis into two cases.

\paragraph{Case 1.}
We consider any cut $(S, V \setminus S)$ such that \[\vol(S) \leq \vol(V \setminus S) \ \ \ \text{and} \ \ \ \sum_{i \in [k] \; : \; \vol(V_i \cap S)  < \frac{2\vol(V_i)}{3}} \vol(V_i \cap S)  \geq \frac{\vol(S)}{2}.\]
In this case, we have
\begin{align*}
 |E(S, V \setminus S)| 
 &\geq
 \sum_{i \in [k] \; : \; \vol(V_i \cap S)  < \frac{2\vol(V_i)}{3}} |E(S \cap V_i, (V \setminus S)\cap V_i)|\\
 &\geq 
 \sum_{i \in [k] \; : \; \vol(V_i \cap S)  < \frac{2\vol(V_i)}{3}} \phi \cdot \min\left\{\vol_{G[V_i]}(S \cap V_i), \vol_{G[V_i]}((V\setminus S) \cap V_i)\right\}\\
 &\geq 
 \sum_{i \in [k] \; : \; \vol(V_i \cap S)  < \frac{2\vol(V_i)}{3}} \phi \cdot \frac{1}{d} \cdot \min\left\{\vol(S \cap V_i), \vol((V\setminus S) \cap V_i)\right\}\\
  &\geq 
 \sum_{i \in [k] \; : \; \vol(V_i \cap S)  < \frac{2\vol(V_i)}{3}} \phi \cdot \frac{1}{2d} \cdot \vol(S \cap V_i)\\
  &\geq 
 \phi \cdot \frac{1}{2d} \cdot \frac{\vol(S)}{2}\\
   &=  \frac{\phi}{4d} \cdot  \vol(S),
\end{align*}
which implies $\Phi(S) \geq \frac{\phi}{4d} > \frac{\epsilon \phi}{6d}$. 

The third inequality in the above calculation is explained as follows:
The condition $\vol(V_i \cap S)  < \frac{2\vol(V_i)}{3}$ implies $|V_i| > 1$, so for each $v \in V_i$, we have $|E(\{v\}, V_i)|\geq \frac{1}{d} \cdot \deg(v)$, which implies 
\[\vol_{G[V_i]}(S \cap V_i) = \sum_{v \in S \cap V_i} |E(\{v\}, V_i)|\geq  \frac{1}{d} \cdot \sum_{v \in S \cap V_i} \deg(v) = \frac{1}{d} \cdot  \vol(S \cap V_i),
\] 
and similarly we have $\vol_{G[V_i]}((V \setminus S) \cap V_i) \geq \frac{1}{d} \cdot  \vol((V \setminus S) \cap V_i)$.

\paragraph{Case 2.}
We consider any cut $(S, V \setminus S)$ such that \[\vol(S) \leq \vol(V \setminus S) \ \ \ \text{and} \ \ \ \sum_{i \in [k] \; : \; \vol(V_i \cap S)  < \frac{2\vol(V_i)}{3}} \vol(V_i \cap S) \leq \frac{\vol(S)}{2}.\]
We construct another cut  $(S', V \setminus S')$ from $(S, V \setminus S)$ as follows: Initially $S' = \emptyset$, and then for each $i \in [k]$, if $\vol(V_i \cap S) \geq \frac{2\vol(V_i)}{3}$, then we add all vertices in $V_i$ to $S'$.
Observe that
\[ \frac{3 \vol(S)}{2}\geq \vol(S') \geq \sum_{i \in [k] \; : \; \vol(V_i \cap S)  \geq \frac{2\vol(V_i)}{3}} \vol(V_i \cap S) \geq \frac{\vol(S)}{2},\]
which implies
\begin{align*}
  \min\left\{\vol(S'),\vol(V \setminus S')\right\}
  &\geq 
  \min\left\{\frac{\vol(S)}{2},\vol(V \setminus S) -\frac{\vol(S)}{2} \right\}\\
  &\geq 
  \min\left\{\frac{\vol(S)}{2}, \frac{\vol(V \setminus S)}{2} \right\}\\ 
  &= \frac{1}{2} \cdot \min\left\{\vol(S),\vol(V \setminus S)\right\}.  
\end{align*}

Therefore, $S' \neq \emptyset$ and $S' \neq V$, so the observation at the beginning of the proof implies that $\Phi(S') \geq \epsilon$. For the rest of the proof, we will relate $\Phi(S)$ and $\Phi(S')$ by upper bounding $|E(S', V \setminus S')|$ in terms of $|E(S, V \setminus S)|$. Specifically, we may upper bound the size of $E(S', V \setminus S')$ by
\[|E(S, V \setminus S)| + 
   \sum_{i \in [k] \; : \; \vol(V_i \cap S)  < \frac{2\vol(V_i)}{3}}  \vol(V_i \cap S) +
    \sum_{i \in [k] \; : \; \vol(V_i \cap S)  \geq \frac{2\vol(V_i)}{3}}  \vol(V_i \setminus S).\]
The reason is that each edge in $E(S', V \setminus S') \setminus E(S, V \setminus S)$ must be incident to a vertex in $V_i \cap S$ for some $i \in [k]$ such that $\vol(V_i \cap S)  < \frac{2\vol(V_i)}{3}$ or incident to a vertex in $V_i \setminus S$ for some $i \in [k]$ such that $\vol(V_i \cap S)  \geq \frac{2\vol(V_i)}{3}$.

For each $i \in [k]$ such that $\vol(V_i \cap S)  < \frac{2\vol(V_i)}{3}$, we must have $|V_i| > 1$, so we may  bound $\vol(V_i \cap S)$ as follows:
\begin{align*}
   \vol(V_i \cap S)
   &\leq
   2 \cdot \min\{\vol(V_i \cap S), \vol(V_i \setminus S)\}\\
   &\leq
   2d \cdot \min\{\vol_{G[V_i]}(V_i \cap S), \vol_{G[V_i]}(V_i \setminus S)\}\\
   &\leq \frac{2d}{\phi} \cdot |E(V_i \cap S, V_i \setminus S)|.
\end{align*}

Next, consider any $i \in [k]$ such that $\vol(V_i \cap S)  \geq \frac{2\vol(V_i)}{3}$. If $|V_i| = 1$, then the unique vertex in $V_i$ must be in $S$, so $\vol(V_i \setminus S) = 0$. Otherwise, 
we may  bound $\vol(V_i \setminus S)$ as follows:
\begin{align*}
   \vol(V_i \setminus S)
   &=
   \min\{\vol(V_i \cap S), \vol(V_i \setminus S)\}\\
   &\leq
   d \cdot \min\{\vol_{G[V_i]}(V_i \cap S), \vol_{G[V_i]}(V_i \setminus S)\}\\
   &\leq \frac{d}{\phi} \cdot |E(V_i \cap S, V_i \setminus S)|.
\end{align*}

Combining these bounds into the above upper bound of $|E(S', V \setminus S')|$, we obtain
\begin{align*}
   |E(S', V \setminus S')| 
   &\leq |E(S, V \setminus S)| + \sum_{i \in [k]} 2d \phi \cdot |E(V_i \cap S, V_i \setminus S)|\\
   &\leq |E(S, V \setminus S)| \cdot \left(1+\frac{2d}{\phi}\right)\\
   &\leq |E(S, V \setminus S)| \cdot \frac{3d}{\phi}.
\end{align*}
Now we are ready to calculate $\Phi(S)$.
\begin{align*}
\Phi(S) &= 
\frac{|E(S, V \setminus S)|}{\min\left\{\vol(S),\vol(V \setminus S)\right\}}\\
&\geq 
\frac{|E(S', V \setminus S')|}{\frac{3d}{\phi}} \cdot \frac{1}{2\min\left\{\vol(S'),\vol(V \setminus S')\right\}}\\
&=
\frac{\phi}{6d} \cdot \Phi(S')\\
&\geq 
\frac{\epsilon \phi}{6d}.
\end{align*}

We conclude that $\Phi(G) \geq \frac{\epsilon \phi}{6d}$, as we have $\Phi(S) \geq \frac{\epsilon \phi}{6d}$ for all $S \subseteq V$ with $S \neq \emptyset$ and $S \neq V$.
\end{proof}

We now describe the algorithm for \cref{lem:existence-variant-aux}.
Let $G=(V,E)$ be an $H$-minor-free graph.
At the beginning of the algorithm, we are given an $\left(\epsilon, \phi, c\right)$ expander decomposition: $\VV= \{V_1, V_2, \ldots, V_k\}$. Our goal is to turn this decomposition into an improved expander decomposition with parameters $\epsilon' = \epsilon \cdot \left(1-\frac{1}{32\alpha}\right)$, $\phi' = \phi \cdot \frac{\epsilon}{13056\alpha^2 c(c+1)}$, and $c' = c+1$.

 \begin{description}
    \item[Step 1: creating singleton clusters.]  For each cluster $S \in \VV$ with $|S| > 1$, and for each vertex $u \in S$ such that $\deg_{G_S}(u) \leq \frac{1}{34\alpha} \cdot \deg_G(u)$, remove $u$ from $S$ and create a new singleton cluster $\{u\}$ associated with the subgraph $G[\{u\}]$ of $G$ induced by $\{u\}$.
    \item[Step 2: creating heavy stars.] Run the heavy-stars algorithm in the cluster graph, where the weight of each edge $\{V_i, V_j\}$ between two adjacent clusters $V_i \in \VV$ and $V_j \in \VV$ in the cluster graph equals $|E(V_i, V_j)|$. Let $\QQ$ denote the set of vertex-disjoint stars computed by the heavy-stars algorithm. For each star $Q \in \QQ$, we write $\VV_Q$ to denote the set of clusters in $Q$ and write $C_Q \in \VV_Q$ to denote the center of $Q$.
    \item[Step 3: removing light links.] We modify each star $Q \in \QQ$, as follows.   For each $S \in \VV_Q \setminus \{C_Q\}$ such that $|E(S, C_Q)| \leq \frac{\epsilon}{64\alpha (c+1)} \cdot \vol(V(G_S))$, remove $S$ from the star $Q$. We emphasize that the volume of $V(G_S)$ is measured in the original graph $G$ and not in  the subgraph $G_S$.
    \item[Step 4: contracting stars.] For each star $Q \in \QQ$, merge the clusters $\VV_Q$ into a new cluster, as follows. 
    The vertex set of the new cluster is $\bigcup_{S \in \VV_Q} S$. The subgraph associated with the new cluster is the union of the subgraphs $G_S$ over all $S \in \VV_Q$, together with all the inter-cluster edges between the clusters in $\VV_Q$.
\end{description}

We write $\VV' = \{V_1', V_2', \ldots, V_{k'}\}$ to denote the expander decomposition computed by the above algorithm, where $k'$ is the number of clusters.
We show that the number of inter-cluster edges is indeed  at most $\epsilon' = \epsilon \cdot \left(1-\frac{1}{32\alpha}\right)$ fraction of the set of all edges $E$.

\begin{lemma}\label{lem:inter-cluster-edges}
    The number of inter-cluster edges   $\frac{1}{2}\sum_{i=1}^{k'} |\partial(V_i')|$ is at most   $\epsilon \cdot \left(1-\frac{1}{32\alpha}\right) \cdot |E|$.
\end{lemma}
\begin{proof}
   We claim that the creation of the singleton clusters in Step~1 only increases the number of inter-cluster edges by at most $\epsilon \cdot |E| \cdot \frac{1}{16\alpha}$. To see this, whenever we make $u \in V$ a singleton cluster, we create at most $\deg(u) \cdot \frac{1}{34\alpha}$ new inter-cluster edges. We charge the cost uniformly to the existing inter-cluster edges incident to $u$. Since the number of existing inter-cluster edges incident to $u$ is at least  $\deg(u) \cdot \left(1 - \frac{1}{34\alpha}\right)$, each of them is charged a cost of at most $\frac{\frac{1}{34\alpha}}{1 - \frac{1}{34\alpha}} < \frac{1}{32\alpha}$, as $\alpha \geq 1$. 
   As each inter-cluster edge is charged at most twice, the total number of new inter-cluster edges created in Step~1 is at most $2 \cdot \epsilon \cdot |E| \cdot \frac{1}{32\alpha} = \epsilon \cdot |E| \cdot \frac{1}{16\alpha}$.

   By \cref{lem:heavystar-weight}, the vertex-disjoint stars constructed by the heavy-stars algorithm in Step~2 capture at least $\frac{1}{8 \alpha}$ fraction of the inter-cluster edges, so the number of these edges is at least  $\epsilon \cdot |E| \cdot \frac{1}{8\alpha}$.
   In Step~3, some of these edges are removed from the stars. The total number of edges removed is at most
\[\sum_{S \in \tilde{\VV}}\frac{\epsilon}{64\alpha (c+1)} \cdot \vol(V(G_S)) \leq \frac{\epsilon}{64\alpha (c+1)} \cdot (c+1) \cdot \vol(V) =  \epsilon \cdot |E| \cdot \frac{1}{32\alpha},\] where $\tilde{\VV}$ refers to the expander decomposition at the beginning of Step~3. The inequality in the above calculation follows from the fact that each vertex $v \in V$ belongs to  $V(G_S)$ for at most $c+1$ clusters $S \in \tilde{\VV}$, as $v$ belongs to at most $c$ clusters in the initial expander decomposition, and this number can go up by at most one after Step~1 of the algorithm.

  Therefore, at the beginning of Step~4, the number of inter-cluster edges in the stars is at least $\epsilon \cdot |E| \cdot \left(\frac{1}{8\alpha} - \frac{1}{32\alpha}\right) = \epsilon \cdot |E| \cdot \frac{3}{32\alpha}$. 
   After the contraction of all stars in Step~4, the total number of inter-cluster edges is at most $\epsilon \cdot |E| \cdot \left(1 + \frac{1}{16\alpha} - \frac{3}{32\alpha}\right) = \epsilon \cdot |E| \cdot \left(1-\frac{1}{32\alpha}\right)$, as claimed.
\end{proof}

We show that for each cluster $S \in \VV'$  with $E(G_S) \neq \emptyset$  in the expander decomposition computed by the algorithm, the subgraph $G_S$ has  conductance at least $\phi' = \phi \cdot \frac{\epsilon}{13056\alpha^2 c(c+1)}$.

\begin{lemma}\label{lem:conductance}
   For each cluster $S \in \VV'$  with $E(G_S) \neq \emptyset$, we have $\Phi(G_S) \geq \phi \cdot \frac{\epsilon}{13056\alpha^2 c(c+1)}$.
\end{lemma}
\begin{proof}
    Each $S \in \VV'$ belongs to one of the following three cases. The first case is where $S$ is a singleton cluster created in Step~1 of the algorithm, in which case we know that $E(G_S) = \emptyset$.
    The second case is where $G_S$ is a subgraph associated with some old cluster in the initial expander decomposition $\VV$, in which case we already know that either $|V(G_S)| = 1$ or $\Phi(G_S) \geq \phi$. For the rest of the proof, we focus on  the third case, where $S$ is a new cluster resulting from merging a star $Q \in \QQ$ in Step~4 of the algorithm.  

    Given the star  $Q \in \QQ$ mentioned above, we construct the graph $G^\ast$ as follows. The construction begins with the disjoint union of the graphs $G_{S'}$, ranging over all $S' \in \VV_Q$, where we treat each $G_{S'}$   as a graph with a distinct vertex set and not a subgraph of $G$, and then we add all the inter-cluster edges between the clusters in $\VV_Q$ in $G$ to connect the graphs $\{G_{S'}\}_{S' \in \VV_Q}$. 
    
    We let $\tilde{G}$ denote the result of merging the vertices in $G^\ast$ corresponding to the same vertex in $G$ and also merging the edges in $G^\ast$ corresponding to the same edge in $G$. Observe that the graph $\tilde{G}$ is identical to $G_S$, where $S$ is the cluster resulting from merging the clusters in $\VV_Q$.

We write $\VV_Q = \{S_1, S_2, \ldots, S_x\}$, where $x = |\VV_Q|$, in such a way that $S_1 = C_Q$ is the center of $Q$.
  We apply \cref{lem:merge-expander-aux} to $G^\ast$ with the partition $V(G^\ast) = V_1^\ast \cup V_2^\ast \cup \cdots \cup V_x^\ast$, where 
 $V_i^\ast=V(G_{S_i})$, to infer that $\Phi(G^\ast) \geq \frac{\epsilon^\ast \phi^\ast}{6d^\ast} =\phi \cdot \frac{\epsilon}{13056\alpha^2 (c+1)}$, as we may use the parameters $\phi^\ast = \phi$, $\epsilon^\ast = \frac{\epsilon}{64\alpha (c+1)}$, and $d^\ast = 34\alpha$. We explain the validity of the choice of parameters, as follows. 
 
 The validity of $\phi^\ast = \phi$ comes from the fact $G_{S_i}$ is a $\phi$-expander if $|V_i^\ast| = |V(G_{S_i})| > 1$.
 Note that the subgraph of $G^\ast$ induced by $V_i^\ast$ is precisely the graph $G_{S_i}$.

To see the validity of the choice of parameter $\epsilon^\ast = \frac{\epsilon}{64\alpha (c+1)}$, consider any $i \in [2,x]$.
 Step~3 of the algorithm ensures that  the number of edges connecting $V_1^\ast$ and $V_i^\ast$ in $G^\ast$ is
\[|E_{G^\ast}(V_1^\ast, V_i^\ast)| 
= 
|E_G(S_1, S_i)| 
> 
\frac{\epsilon}{64\alpha (c+1)} \cdot \vol_G(V(G_{S_i})) 
\geq 
\frac{\epsilon}{64\alpha (c+1)} \cdot \vol_{G^\ast}(V_i^\ast).
\]

  For the validity of $d^\ast = 34\alpha$, consider any $i \in [x]$ with $|V_i^\ast| = |V(G_{S_i})| > 1$.
    Step~1 of the algorithm ensures that  each $v \in S_i \subseteq V_i^\ast$ has at least $\frac{1}{34\alpha} \cdot \deg_{G}(v) \geq \frac{1}{34\alpha} \cdot \deg_{G^\ast}(v)$ incident edges in $G_{S_i}$. 
    Recall again that $G_{S_i}$ is identical to the subgraph of $G^\ast$ induced by $V_i^\ast$, so the above calculation implies that $|E_{G^\ast}(\{v\}, V_i^\ast)| \geq \frac{1}{34\alpha} \cdot \deg_{G^\ast}(v)$.
For the case $v \in V_i^\ast \setminus S_i$, our construction of $G^\ast$ ensures that all neighbors of $v$ are in $V_i^\ast$, so trivially $|E_{G^\ast}(\{v\}, V_i^\ast)| = \deg_{G^\ast}(v)$. 

    To finish the proof of the lemma, we just need to show that $\Phi(\tilde{G}) \geq \frac{1}{c} \cdot\Phi(G^\ast)$. This follows from the observation that each edge in $\tilde{G}$ is the result of merging at most $c$ edges in $G^\ast$, so each cut $(\tilde{C}, V(\tilde{G}) \setminus \tilde{C})$ of $\tilde{G}$ naturally corresponds to a cut $(C^\ast, V(G^\ast) \setminus C^\ast)$ of $G^\ast$ with $\Phi(\tilde{C}) \geq \frac{1}{c} \cdot \Phi(C^\ast)$.
    The reason that each edge in $\tilde{G}$ is the result of merging at most $c$ edges in $G^\ast$ is that each vertex $u \in V$ belongs to at most $c$ subgraphs in $\{G_{S'}\}_{S' \in \VV}$ such that $E(G_{S'}) \neq \emptyset$, as each cluster $S'$ created in Step~1 of the algorithm must have $E(G_{S'}) = \emptyset$.  
\end{proof}

We are now ready to prove \cref{lem:existence-variant-aux}.

\begin{proof}[Proof of \cref{lem:existence-variant-aux}]
    The validity of the conductance bound $\phi'$ and the bound on the number of inter-cluster edges $\epsilon'$ follows from \cref{lem:inter-cluster-edges,lem:conductance}. For the validity of the bound $c' = c+1$, we observe that each vertex $v \in V$ belongs to  $V(G_S)$ for at most $c+1$ clusters $S \in \VV'$, as $v$ belongs to at most $c$ clusters in the initial expander decomposition, and this number can go up by at most one after Step~1 of the algorithm. 
\end{proof}

\paragraph{Distributed implementation.}  We show that the decomposition of \cref{lem:existence-variant} can be computed in $2^{O\left(\log^2 \frac{1}{\epsilon}\right)} \cdot  O(\log \Delta)  \cdot \left(O(\log^3 \Delta) +O(\log^\ast n)\right)$ rounds deterministically in $\CONGEST$ using the routing algorithm of \cref{lem:gathering-1}. 
It is well-known that an $s$-vertex $\phi$-expander has diameter $O\left(\frac{\log s}{\phi}\right)$. \cref{lem:separator} implies that 
 the number of vertices in a $\phi$-expander with maximum degree $\Delta$ is at most $s = O\left(\frac{\Delta}{\phi^2}\right)$, so such a graph has diameter at most $O\left(\frac{\log s}{\phi}\right) = O\left(\frac{\log \Delta + \log \frac{1}{\phi}}{\phi}\right)$. 

Steps~1,~3,~and~4 of the algorithm of \cref{lem:existence-variant-aux} can all be implemented to take $O(cD)$ rounds deterministically in a straightforward manner, where $D = O\left(\frac{\log \Delta + \log \frac{1}{\phi}}{\phi}\right)$ is an upper bound of the diameter of $G_S$, for each cluster $S \in \VV$ in the initial expander decomposition.  The factor of $c$ reflects the congestion due to the fact that each edge $e\in E$ can belong to at most $c$ distinct subgraphs $G_S$. 

 Step~2 of the algorithm of \cref{lem:existence-variant-aux} requires running the  heavy-stars algorithm in a cluster
graph. As discussed earlier, all steps in the heavy-stars algorithm can be implemented to take $O(cD) \cdot O(\log^\ast n)$ rounds deterministically, except for the first step where each cluster $S \in \VV$ needs to identify a cluster $S' \in \VV \setminus \{S\}$ that maximizes $|E(S, S')|$.  Again, we need a factor of $c$ here to deal with the congestion. 

This task can be solved by information gathering, as follows. Let $\vstar$ be a vertex in $G_S$ that has the highest degree. Each vertex $v \in V(G_S)$ send  $(\ID(S'), |E(\{v\}, S')|)$, for all  clusters $S' \in \VV \setminus \{S\}$ such that  $E(\{v\}, S') \neq \emptyset$ to $\vstar$. After that, $\vstar$ can locally identify a cluster $S' \in \VV \setminus \{S\}$ that maximizes $|E(S, S')|$. The information each $v \in V(G_S)$ needs to send can be encoded as $O(\deg_G(v))$ messages of $O(\log n)$ bits. 

Step~1 of the algorithm of \cref{lem:existence-variant-aux}  ensures that $\deg_G(v) = O(\deg_{G_S}(v))$ at the beginning of Step~2, so we may solve the above routing task using the routing algorithm of  \cref{lem:gathering-1}, where we set $f = \frac{1}{2|E(G_S)|+1}$ to ensure that all messages are delivered.
As discussed in \cref{sect:routing-minor-free}, the round complexity of the routing algorithm of  \cref{lem:gathering-1} is $O\left(\frac{\log^3 |E(G_S)|}{\phi^4}\right)$, where $\log |E(G_S)| = \left(\log \Delta + \log \frac{1}{\phi}\right)$.
To summarize, Steps~1,~2,~3,~and~4 of the above algorithm  cost \[O(cD) \cdot O\left(\log^\ast n + \frac{\left(\log \Delta + \log \frac{1}{\phi}\right)^3}{\phi^4}\right) = \poly\left(\frac{1}{\phi}\right) \cdot O(c) \cdot O(\log \Delta)  \cdot \left(O(\log^3 \Delta) +O(\log^\ast n)\right)\] rounds. The algorithm of \cref{lem:existence-variant} involves running this algorithm for $O\left(\log\frac{1}{\epsilon
}\right)$ iterations, with $c = O\left(\log\frac{1}{\epsilon
}\right)$ and $\phi = 2^{-O\left(\log^2 \frac{1}{\epsilon}\right)}$, so the overall round complexity for computing the decomposition of \cref{lem:existence-variant} is
\[2^{O\left(\log^2 \frac{1}{\epsilon}\right)} \cdot  O(\log \Delta)  \cdot \left(O(\log^3 \Delta) +O(\log^\ast n)\right).\]

\section{Our decomposition algorithm}\label{sect:decomposition}

The main goal of this section is to prove \cref{lem-algo-main} by developing efficient deterministic algorithms for computing an $(\epsilon, D, T)$-decomposition in $H$-minor-free networks. Instead of relying on the
 distributed algorithm of \cref{lem:existence-variant} described in \cref{sect:ldd-modify}, we present a  \emph{more efficient} method that only utilizes \cref{lem:existence-variant} in an \emph{existential} fashion.

\paragraph{Existential results.} In \cref{sect:routing-exist}, we first show that certain $(\epsilon, D, T)$-decompositions exist in $H$-minor-free networks, by combining the results developed in \cref{sect:routing,sect:existential,sect:existence-variant}. Intuitively, combining the routing algorithm of \cref{lem:gathering-1} or \cref{lem:gathering-3} with an $(\epsilon, \phi)$ expander decomposition  yields an $(\epsilon, D, T)$-decomposition with $D = O\left(\frac{\log \Delta + \log \frac{1}{\phi}}{\phi}\right)$ and $T = \poly\left(\frac{1}{\phi}, \log \Delta\right)$. We will show that the parameters $D$ and $T$ in the $(\epsilon, D, T)$-decomposition can be improved by further partitioning each cluster in the expander decomposition and only requiring the routing algorithm to work for a fraction of the vertices.

\paragraph{Algorithms.} In \cref{sect:routing-algo}, we design  algorithms for computing an $(\epsilon, D, T)$-decomposition in $H$-minor-free networks by combining the existential results with the heavy-stars algorithm of~\cite{czygrinow2008fast}. The way we use the  heavy-stars algorithm  here is different from how we use the heavy-stars algorithm in \cref{sect:ldd-modify}. When we merge a star $Q$ to form a new cluster, here we do not try to show any conductance bound for the subgraph associated with the new cluster. Instead, we simply use the routing algorithm $\AAA$ associated with the given decomposition to gather the entire information about the graph topology of $G[S]$, for all clusters $S \in \VV_Q$ in the star,  to the leader $\vstar_{C_Q}$ of the cluster $C_Q$ that is the center of the star. The leader $\vstar_{C_Q}$ then locally computes the best possible decomposition for the subgraph of $G$ induced by $\bigcup_{S \in \VV_Q} S$. After that, $\vstar_{C_Q}$ can communicate with the vertices in $\bigcup_{S \in \VV_Q} S$ to let them learn the decomposition. The efficiency of an algorithm based on this approach requires the length of the bit string $B_v$ to be small in the decomposition.

\subsection{Existential results}\label{sect:routing-exist}

The following lemma is proved by combining the $(\epsilon,\phi,c)$ expander decomposition of \cref{lem:existence-variant}, the routing algorithm of \cref{lem:gathering-1}, and 
 the low-diameter decomposition of~\cref{lem:LDD-basic}.

\begin{lemma}\label{lem:existential-1}
For any $\epsilon \in \left(0, \frac{1}{2}\right)$, any $H$-minor-free graph $G=(V,E)$ with maximum degree $\Delta$ admits  an $(\epsilon, D, T)$-decomposition with $D = O\left(\frac{1}{\epsilon}\right)$ and $T = 2^{O\left(\log^2 \frac{1}{\epsilon}\right)} \cdot  O\left(\log \Delta \right)$. The length of the bit string $B_v$ is $O\left(\log \frac{1}{\epsilon}\right) \cdot \left(O(\log n) +\deg(v)\right)$ for each $v \in V$.
\end{lemma}
\begin{proof}
Let $\VV=\{V_1, V_2, \ldots, V_k\}$ be an $\left(\frac{\epsilon}{4}, \phi, c\right)$ expander decomposition of $G$ with $\phi = 2^{-O\left(\log^2 \frac{1}{\epsilon}\right)}$ and $c = O\left(\log \frac{1}{\epsilon}\right)$, whose existence is guaranteed by \cref{lem:existence-variant}. We construct an $(\epsilon, D, T)$-decomposition by modifying this expander decomposition.

\paragraph{Step~1: creating singleton clusters.}
Similar to Step~1 of the algorithm of \cref{lem:existence-variant-aux}, for each cluster $S \in \VV$ with $|S| > 1$, and for each vertex $u \in S$ such that $\deg_{G_S}(u) \leq \frac{1}{4} \cdot \deg_G(u)$, remove $u$ from $S$ and create a new singleton cluster $\{u\}$ associated with the subgraph $G[\{u\}]$ of $G$ induced by $\{u\}$. The result is an $\left(\frac{\epsilon}{2}, \phi, c\right)$ expander decomposition, as the modification at most doubles the number of inter-cluster edges. This can be proved using the charging argument  in the proof of \cref{lem:inter-cluster-edges}. The purpose of this step is to ensure that for each vertex $u$ in each cluster $S$ with $|S|>1$, at least $1/4$ of the edges incident to $u$ are included in $G_S$.

\paragraph{Step~2: routing.}
In the modified expander decomposition, for each cluster $S \in \VV$ with $|S| > 1$, run the routing algorithm of \cref{lem:gathering-1} in $G_S$ with $f = \frac{\epsilon}{16c}$, where $\vstar$ is selected as any vertex in $G_S$ that has the highest degree in $G_S$. The routing algorithm costs 
\[O\left(\frac{\frac{|E(G_S)|}{\Delta(G_S)} \cdot \log |E(G_S)| \cdot \log^2 \frac{1}{f}}{\phi^2}\right) = 
O\left( \frac{\left(\log \Delta + \log \frac{1}{\phi}\right)  \cdot \log^2 \frac{1}{f}}{\phi^4}\right)
= 2^{O\left(\log^2 \frac{1}{\epsilon}\right)} \cdot  O\left(\log \Delta \right)
\]
rounds. In the above calculation, we use the fact that $\frac{|E(G_S)|}{\Delta(G_S)}=O\left(\frac{1}{\phi^2}\right)$, which is due to~\cref{lem:separator} and the fact that $|E(G_S)|= O(|V(G_S)|)$, as $H$-minor-free graphs have bounded arboricity. We also have $|E(G_S)| = O\left( \frac{\Delta}{\phi^2} \right)$, as $\Delta(G_S)$  is at most the maximum degree $\Delta$ of $G$.

In the routing algorithm, each vertex $v \in V(G_S)$ wants to 
send $\deg_{G_S}(v)$ messages of $O(\log n)$ bits  to $\vstar$, and it is guaranteed that at least $1-f$ fraction of these  messages are successfully delivered. We let $F_S$ denote the set of vertices $u$ in $S$ such that at least $\frac{\deg_{G_S}(u)}{2}$ messages of $u$ are not delivered. By the correctness of the routing algorithm, \[\frac{1}{2} \cdot \vol_{G_S}(F_S) \leq f \cdot 2|E(G_S)| \leq  f \cdot\vol_G(V(G_S)),\]
as $\frac{1}{2} \cdot \vol_{G_S}(F_S)$ is a lower bound on the number of messages that are not delivered.
Since each vertex belongs to $V(G_S)$ for at most $c$ distinct clusters $S$, we infer that
\[\sum_{S \in \VV} \vol_{G[S]}(F_S) \leq
2 \cdot \sum_{S \in \VV}  f \cdot \vol_G(V(G_S)) \leq f \cdot c \cdot 4|E| \leq \frac{\epsilon}{4} \cdot |E|.
\]
We further modify the clustering $\VV$ by making each $u \in F_S$ a singleton cluster $\{u\}$ associated with the subgraph $G[\{u\}]$ and removing all vertices in $F_S$ from $S$, for each cluster $S \in \VV$. In view of the above calculation, the increase in the number of inter-cluster edges, which is upper bounded by $\sum_{S \in \VV} \vol_{G[S]}(F_S)$, is at most $\frac{\epsilon}{4} \cdot |E|$. 

The purpose of removing $F_S$ from $S$ is to ensure that the above routing algorithm is able to deliver at least $\frac{\deg_{G_S}(u)}{2} \geq \frac{\deg_{G}(u)}{8}$ messages from each remaining vertex $u$ in the cluster $S$.

\paragraph{Step~3: diameter reduction.} Finally, for each cluster $S \in \VV$ in the modified expander decomposition, we apply the low-diameter decomposition of~\cref{lem:LDD-basic} to $G[S]$ with parameter $\frac{\epsilon}{4}$. After that, each cluster has diameter $O\left(\frac{1}{\epsilon}\right)$, and the increase in the number of inter-cluster edges is at most $\frac{\epsilon}{4} \cdot |E|$.

\paragraph{Summary.} We verify that the above construction gives a valid $(\epsilon, D, T)$-decomposition and describe how the routing is performed. The number of inter-cluster edges is at most $\epsilon|E|$ because the initial expander decomposition has at most  $\frac{\epsilon}{4} \cdot |E|$ inter-cluster edges and each of the three steps  increases the number by at most $\frac{\epsilon}{4} \cdot |E|$. The diameter bound $D = O\left(\frac{1}{\epsilon}\right)$ follows from the description of Step~3.

For the rest of the proof, we describe the routing algorithm $\AAA$ in the definition of our decomposition.
In the subsequent discussion, we write $\VV'$ to denote the final clustering and write $\VV$ to denote the clustering at the beginning of Step~2. Each $S' \in \VV'$ is a subset of some cluster $S \in \VV$. If $|S'| = 1$, then nothing needs to be done. Otherwise, we set the leader $\vstar_{S'}$ of $S'$ to be the vertex $\vstar \in S$ considered in Step~2.
It is required that the algorithm $\AAA$  sends $\deg_G(v)$ messages of $O(\log n)$ bits from each vertex $v \in S'$  to   $\vstar_{S'} = \vstar$. To do so, we just need to run the routing algorithm in Step~2 eight times, as it is guaranteed that in each execution of the routing algorithm, at least $\frac{\deg_{G_S}(v)}{2} \geq \frac{\deg_{G}(v)}{8}$ messages from $v$ are successfully delivered. Since we need to run the routing algorithm in parallel for all $S \in \VV$, there is a $c$-factor overhead in the round complexity, so the round complexity of $\AAA$ is $c \cdot 2^{O\left(\log^2 \frac{1}{\epsilon}\right)} \cdot  O\left(\log \Delta \right) = 2^{O\left(\log^2 \frac{1}{\epsilon}\right)} \cdot  O\left(\log \Delta \right) = T$, as required.

We show that in order to run $\AAA$, the  information needed to store at each 
 vertex $v \in V$ costs only $O\left(\log \frac{1}{\epsilon}\right) \cdot \left(O(\log n) +\deg(v)\right)$ bits. In order to run the routing algorithm of Step~2, all we need is that each vertex $v$ in $G_S$ knows the list of incident edges in $G_S$. As $v$ belongs to $G_S$ for at most $c$ distinct clusters $S$, storing this information for all $S \in \VV$ costs $c \cdot \left(O(\log n) +\deg(v)\right) = O\left(\log \frac{1}{\epsilon}\right) \cdot \left(O(\log n) +\deg(v)\right) = |B_v|$ bits, where the term $O(\log n)$ is the cost of storing the identifier of cluster $S$ and the identifier of the leader $\vstar$ of $G_S$.  
\end{proof}

The reason that we use \cref{lem:gathering-1} and not \cref{lem:gathering-3} to prove \cref{lem:existential-1} is that 
if we replace \cref{lem:gathering-1} with \cref{lem:gathering-3} in the proof, then 
the $\log \log \zeta$ term in the round complexity of \cref{lem:gathering-3} will lead to an extra $O(\log \log \Delta)$ factor in the round complexity.

The following lemma is proved by combining the expander decomposition of \cref{fact:existence-basic}, the routing algorithm of \cref{lem:gathering-3}, and 
 the low-diameter decomposition of~\cref{lem:LDD-basic}.

\begin{lemma}\label{lem:existential-2}
For any $\epsilon \in \left(0, \frac{1}{2}\right)$, any $H$-minor-free graph $G=(V,E)$ admits  an $(\epsilon,  D, T)$-decomposition with $D = O\left(\frac{1}{\epsilon}\right)$ and $T = O\left(\frac{\log^5 \Delta \log \frac{1}{\epsilon} +   \log^6 \frac{1}{\epsilon}}{\epsilon^4}\right)$.  
For each $v \in V$, the bit string $B_v$ has two parts:
\begin{itemize}
    \item  The first part has length $O\left(\frac{\log^5 \Delta \log \frac{1}{\epsilon} +   \log^6 \frac{1}{\epsilon}}{\epsilon^4}\right) \cdot O(\log n)$ and is identical for all vertices in $V$.
    \item The second part has length $O(\log n)$.
\end{itemize}
\end{lemma}
\begin{proof}
Let $\VV=\{V_1, V_2, \ldots, V_k\}$ be an $\left(\frac{\epsilon}{4}, \phi\right)$ expander decomposition of $G$ with $\phi = \Omega\left( \frac{\epsilon}{\log |V|}\right) = \Omega\left( \frac{\epsilon}{\log \Delta + \log \frac{1}{\epsilon}}\right)$, whose existence is guaranteed by \cref{fact:existence-basic}. Similar to the proof of \cref{lem:existential-1}, we construct an $(\epsilon, D, T)$-decomposition by modifying the expander decomposition in three steps. 

\paragraph{Step~1: creating singleton clusters.}
Same as the proof of \cref{lem:existential-1}, for each cluster $S \in \VV$ with $|S| > 1$, and for each vertex $u \in S$ such that $\deg_{G[S]}(u) \leq \frac{1}{4} \cdot \deg_G(u)$, remove $u$ from $S$ and create a new singleton cluster $\{u\}$ associated with the subgraph $G[\{u\}]$ of $G$ induced by $\{u\}$. The modification causes the number of inter-cluster edges to increase by at most $\frac{\epsilon}{4} \cdot |E|$.

\paragraph{Step~2: routing.}
In the modified expander decomposition, we run the routing algorithm of \cref{lem:gathering-3} in $\{G_1, G_2, \ldots, G_s\} = \{G[S]\}_{S \in \VV \; : \; |S| > 1}$ with parameter $f = \frac{\epsilon}{16}$. In the routing algorithm, for each $i \in [s]$, each vertex $v$ in $G_i = (V_i, E_i)$ wants to 
send $\deg_{{G_i}}(v) \geq \frac{\deg_G(v)}{4}$ messages of $O(\log n)$ bits  to $\vstar_i$, where $\vstar_i$ is selected as any vertex in $G_i$ whose degree in $G_i$ equals the maximum degree $\Delta_i$ of $G_i$.
The routing algorithm guarantees that at least $1-f = 1 - \frac{\epsilon}{16}$ fraction of these messages are successfully delivered.
The routing algorithm costs 
\[O\left(\eta \cdot  \log \frac{1}{f} + \log \phi^{-1} + \log \log \zeta\right) \cdot O(\phi^{-2} \log \zeta)
\]
rounds. In the above calculation, $\eta$  is the maximum value of $\frac{|E_i|}{|\Delta_i|}$ over all $1 \leq i \leq s$, and  $\zeta$ is  the maximum value of $|E_i|$ over all $1 \leq i \leq s$.
Similar to the proof of \cref{lem:existential-1}, we have $\eta = O\left(\frac{1}{\phi^2}\right)$ and 
$\zeta = O\left( \frac{\Delta}{\phi^2} \right)$. Therefore, we may upper bound round complexity in terms of $\epsilon$ and $\Delta$ as follows.
\begin{align*}
 O\left(\eta \cdot  \log \frac{1}{f} + \log \phi^{-1} + \log \log \zeta\right) \cdot O(\phi^{-2} \log \zeta)
&= 
O\left(\eta \cdot  \log \frac{1}{f} \right) \cdot O(\phi^{-2} \log \zeta)\\
&=
O\left(\frac{\log \frac{1}{f} \cdot \log \zeta}{\phi^4}\right)\\
&=
O\left( \log \frac{1}{\epsilon} \cdot \left(\log \Delta + \log \frac{1}{\epsilon}\right) \cdot \frac{\log^4 \Delta + \log^4 \frac{1}{\epsilon} }{\epsilon^4}\right)\\
&=
O\left(\frac{\log^5 \Delta \log \frac{1}{\epsilon} +   \log^6 \frac{1}{\epsilon}}{\epsilon^4}\right).
\end{align*}

For each cluster $S \in \VV$ with $|S| > 1$, we let $F_S$ denote the set of vertices $u$ in $S$ such that at least $\frac{\deg_{G[S]}(u)}{2}$ messages of $u$ are not delivered. By the correctness of the routing algorithm, \[\frac{1}{2} \cdot \vol_{G[S]}(F_S) \leq f \cdot 2|E|,\]  since $\frac{1}{2} \cdot \vol_{G[S]}(F_S)$ is a lower bound on the total number of messages that are not delivered and $2|E|$ is an upper bound on the total number of messages, as the graphs $G[S]$ are disjoint for all clusters $S \in \VV$.

We further modify the clustering $\VV$ by making each $u \in F_S$ a singleton cluster $\{u\}$ associated with the subgraph $G[\{u\}]$ and removing all vertices in $F_S$ from $S$, for each cluster $S \in \VV$. The increase in the number of inter-cluster edges is upper bounded by
\[\sum_{S \in \VV} \vol_{G[S]}(F_S)  
\leq 4f \cdot |E| = \frac{\epsilon}{4} \cdot |E|.
\]

Same as the proof of \cref{lem:existential-1}, the purpose of removing $F_S$ from $S$ is to ensure that the above routing algorithm is able to deliver at least $\frac{\deg_{G[S]}(u)}{2} \geq \frac{\deg_{G}(u)}{8}$ messages from each remaining vertex $u$ in the cluster $S$.

\paragraph{Step~3: diameter reduction.} Same as the proof of \cref{lem:existential-1}, for each cluster $S \in \VV$ in the modified expander decomposition, we apply the low-diameter decomposition of~\cref{lem:LDD-basic} to $G[S]$ with parameter $\frac{\epsilon}{4}$. After that, each cluster has diameter $O\left(\frac{1}{\epsilon}\right)$, and the increase in the number of inter-cluster edges is at most $\frac{\epsilon}{4} \cdot |E|$.

\paragraph{Summary.} We verify that the above construction gives a valid $(\epsilon, D, T)$-decomposition and describe how the routing is performed. The proof of the validity of $\epsilon$ and $D$ are the same as the proof of \cref{lem:existential-1}, so for the rest of the proof, we focus on the routing algorithm.

In the subsequent discussion, we write $\VV'$ to denote the final clustering and write $\VV$ to denote the clustering at the beginning of Step~2.
Consider any $S' \in \VV'$.
If $|S'| = 1$, then nothing needs to be done.
Otherwise $G[S']$ is a subgraph of one $G_i$ of the subgraphs in $\{G_1, G_2, \ldots, G_s\} = \{G[S]\}_{S \in \VV \; : \; |S| > 1}$.
We set the leader $\vstar_{S'}$ of $S'$ to be the vertex $\vstar_i$ considered in Step~2.
It is required that the algorithm $\AAA$  sends $\deg_G(v)$ messages of $O(\log n)$ bits from each vertex $v \in S'$  to   $\vstar_{S'} = \vstar_i$. To do so, we just need to run the routing algorithm in Step~2 eight times, as it is guaranteed that in each execution of the routing algorithm, at least $\frac{\deg_{G[S]}(v)}{2} \geq \frac{\deg_{G}(v)}{8}$ messages from $v$ are successfully delivered.

In view of \cref{lem:gathering-3}, in order to run the above algorithm $\AAA$, we need  to store  
a string of 
\[O\left(\eta \cdot  \log \frac{1}{f} + \log \phi^{-1} + \log \log \zeta\right) \cdot O(\phi^{-2} \log \zeta) \cdot O(\log n) = O\left(\frac{\log^5 \Delta \log \frac{1}{\epsilon} +   \log^6 \frac{1}{\epsilon}}{\epsilon^4}\right) \cdot O(\log n)\]
bits in all vertices in $G_1, G_2, \ldots, G_s$. This bit string encodes the routing schedule. Other than this information, for each subgraph $G_i$ in $\{G_1, G_2, \ldots, G_s\} = \{G[S]\}_{S \in \VV \; : \; |S| > 1}$, each vertex $v$ in $G_i$ just needs to know the identifier of $G_i$ and the identifier of the vertex $\vstar_i$ in order to run the routing algorithm. Storing this information costs $O(\log n)$ bits.
\end{proof} 

The reason that we use \cref{lem:gathering-3} and not \cref{lem:gathering-1} to prove \cref{lem:existential-2} is that the round complexity of \cref{lem:gathering-3} has a better dependence on $\log \frac{1}{f}$. Specifically, replacing \cref{lem:gathering-3} with \cref{lem:gathering-1} in the proof of \cref{lem:existential-2} will cause the round complexity $T$ to increase by a factor of $O(\log \frac{1}{\epsilon})$.

\subsection{Algorithms}\label{sect:routing-algo}

In this section, we design distributed algorithms constructing $(\epsilon, D, T)$-decompositions matching the existential bounds given in \cref{lem:existential-1,lem:existential-2}. 
In the following lemma, we use heavy-stars decomposition to improve the parameter $\epsilon$ by a constant factor at the cost of increasing $D$ and $T$ for any given  $(\epsilon, D, T)$-decomposition of an $H$-minor-free graph.
Same as \cref{sect:existence-variant}, $\alpha=O(1)$ is an upper bound on the arboricity of $H$-minor-free graphs. 

\begin{lemma}\label{lem-algo-1}
    Suppose an $(\epsilon, D, T)$-decomposition of an $H$-minor-free graph $G=(V,E)$ is given. An $(\epsilon', D', T')$-decomposition with $\epsilon' = \epsilon \cdot \left(1 - \frac{1}{16\alpha}\right)$, $D' = 3D+2$, and $T' = O\left(\frac{T+1}{\epsilon}\right)$ can be computed in $O\left((D+1) \log^\ast n + T\right)$ rounds.
\end{lemma}
\begin{proof}
    Let $\VV$ be the clustering of the given $(\epsilon, D, T)$-decomposition. The algorithm for constructing an $(\epsilon', D', T')$-decomposition consists of Steps 2,~3,~and~4 of the algorithm for \cref{lem:existence-variant-aux}, where in Step~3 we change the condition for removing $S$ from $Q$ to be $|E(S, C_Q)| \leq \frac{\epsilon}{32\alpha} \cdot \vol(S)$. In Step~4, when a star $Q \in \QQ$ is merged into a new cluster  $S'$, the leader $\vstar_{S'}$ of $S'$ in the new decomposition is chosen as the existing leader $\vstar_{C_Q}$ of the cluster $C_Q$ that is the center of $Q$ in the given decomposition.

    \paragraph{Round complexity.}
    We verify that the algorithm costs $O\left((D+1) \log^\ast n + T\right)$ rounds.
    Steps~3 and~4 can be implemented to run in $O(D+1)$ rounds in a straightforward manner. The heavy-stars algorithm of Step~2 can be implemented to run in $O\left((D+1) \log^\ast n + T\right)$ rounds, as the four steps of the heavy-stars algorithm can be implemented to run in $O(T)$, $O((D+1) \log^\ast n)$, $O(D+1)$, and $O(D+1)$ rounds, respectively. The reason that the first step takes  $O(T)$ rounds is that we need to use the routing algorithm $\AAA$ for the given 
$(\epsilon, D, T)$-decomposition for each cluster $S \in \VV$ to identify a neighboring cluster $S' \in \VV$ that maximizes $|E(S, S')|$.

\paragraph{Routing algorithm.}
    We need to describe a routing algorithm $\AAA'$  that allows each vertex $v \in S'$ in each cluster $S' \in \VV'$ to send $\deg(v)$ messages of $O(\log n)$ bits to the leader $\vstar_{S'}$. If $S'$ itself is already an existing cluster in $\VV$, then we may just run the existing routing algorithm $\AAA$ for the given 
$(\epsilon, D, T)$-decomposition. Otherwise, $S'$ is the result of  merging a star $Q \in \QQ$, in which case the goal of the routing task is to gather all the messages in $S' = \bigcup_{S \in \VV_Q} S$ to $\vstar_{C_Q}$. To do so, for each cluster $S \in \VV_Q \setminus \{C_Q\}$, we first use the given routing algorithm $\AAA$ to gather all the messages in $S$ to $\vstar_S$. We also use $\AAA$ to let $\vstar_S$ learn $|E(\{v\}, C_Q)|$ for each $v \in S$. After that, $\vstar_S$ redistributes the messages to each 
$v \in S$ in such a way that the number of messages $v \in S$ holds is $|E(S, C_Q)| \cdot  O\left(\frac{1}{\epsilon}\right)$. 
This is possible because $|E(S, C_Q)| > \frac{\epsilon}{32\alpha} \cdot \vol(S)$ due to the new condition in Step~3. The redistribution costs $O\left(\frac{T}{\epsilon}\right)$ rounds using $\AAA$.
After redistributing the messages in each cluster $S \in \VV_Q \setminus \{C_Q\}$, we may transmit all the messages in $\bigcup_{S \in \VV_Q \setminus \{C_Q\}} S$ to vertices in $C_Q$ in $O\left(\frac{1}{\epsilon}\right)$ rounds using the inter-cluster edges. As the load per vertex $v$ in $C_Q$ is $O\left(\frac{1}{\epsilon}\right) \cdot \deg(v)$, we may use $\AAA$ to gather all the messages to $\vstar_{C_Q}$ in $O\left(\frac{T}{\epsilon}\right)$ rounds.

    \paragraph{Validity of parameters.}
    The above discussion shows that the choice of the parameter $T' = O\left(\frac{T+1}{\epsilon}\right)$ is valid.
    To show that the new clustering $\VV'$  is indeed an $(\epsilon', D', T')$-decomposition, we still need to verify the validity of the parameters $\epsilon'$ and $D'$. The validity of $D' = 3D+2$ follows from the fact that the diameter of a star is at most $2$, which implies that the diameter of each new cluster is at most $3D+2$.  The validity of $\epsilon' = \epsilon \cdot \left(1 - \frac{1}{16\alpha}\right)$ can be proved in a way similar to the proof of \cref{lem:inter-cluster-edges}. Specifically, by \cref{lem:heavystar-weight}, the stars computed by the heavy-stars algorithm in Step~2 capture at least $\frac{1}{8\alpha}$ fraction of the inter-cluster edges, and Step~3 can remove at most $\sum_{S \in \VV} \frac{\epsilon}{32\alpha} \cdot \vol(S) \leq \frac{\epsilon}{32\alpha} \cdot 2 |E| \leq \frac{\epsilon}{16\alpha}  |E|$. Therefore, when we merge the stars in Step~4, we reduce the number of inter-cluster edges by at least $\frac{\epsilon}{8\alpha} |E| - \frac{\epsilon}{16\alpha}  |E| = \frac{\epsilon}{16\alpha}  |E|$, so the total number of inter-cluster edges in the new clustering $\VV$ is at most $\epsilon \cdot \left(1 - \frac{1}{16\alpha}\right) |E| = \epsilon'|E|$.    
\end{proof}

In the following two lemmas, we use the information-gathering algorithm $\AAA$ to improve the parameters $D$ and $T$ at the cost of slightly increasing $\epsilon$ by a constant factor, for any given $(\epsilon, D, T)$-decomposition of a $H$-minor-free graph.

\begin{lemma}\label{lem-algo-2}
    Suppose an $(\epsilon, D, T)$-decomposition  of an $H$-minor-free graph $G$ is given. An $(\epsilon', D', T')$-decomposition with $\epsilon' = \epsilon \cdot \left(1 + \frac{1}{32\alpha}\right)$, $D' = O\left(\frac{1}{\epsilon}\right)$, and $T' = 2^{O\left(\log^2 \frac{1}{\epsilon}\right)} \cdot  O\left(\log \Delta \right)$ can be constructed in $O\left(T \cdot \log \frac{1}{\epsilon} \right)$ rounds.
\end{lemma}
\begin{proof}
    We spend $O(T)$ rounds to gather the entire graph topology of $G[S]$ to $\vstar_S$ for each cluster $S \in \VV$ in the given $(\epsilon, D, T)$-decomposition.
    After that, $\vstar_S$ locally computes an $(\epsilon^\ast, D^\ast, T^\ast)$-decomposition $\VV^\ast$ with $\epsilon^\ast = \frac{\epsilon}{32\alpha}$, $D^\ast = O\left(\frac{1}{\epsilon}\right)$, and $T^\ast = 2^{O\left(\log^2 \frac{1}{\epsilon}\right)} \cdot  O\left(\log \Delta \right)$, whose existence is guaranteed by \cref{lem:existential-1}. The length of the bit string $B_v$ is $O\left(\log \frac{1}{\epsilon}\right) \cdot \left(O(\log n) +\deg(v)\right)$ for each $v \in S$, so we can spend $O\left(T \cdot \log \frac{1}{\epsilon} \right)$ rounds to let $\vstar_S$ transmit the bit string $B_v$ and the identifier of the new cluster of $v$ in $\VV^\ast$ to each $v \in S$. 
    The union of the $(\epsilon^\ast, D^\ast, T^\ast)$-decompositions for all $S \in \VV$ is a desired $(\epsilon', D', T')$-decomposition because we only create at most $\frac{\epsilon |E|}{32\alpha}$ new inter-cluster edges.
\end{proof}

\begin{lemma}\label{lem-algo-3}
    Suppose an $(\epsilon, D, T)$-decomposition of an $H$-minor-free graph $G$ is given. An $(\epsilon', D', T')$-decomposition with $\epsilon' = \epsilon \cdot \left(1 + \frac{1}{32\alpha}\right)$, $D' = O\left(\frac{1}{\epsilon}\right)$, and $T' = O\left(\frac{\log^5 \Delta \log \frac{1}{\epsilon} +   \log^6 \frac{1}{\epsilon}}{\epsilon^4}\right)$ can be constructed in $O\left(T + T' + D\right)$ rounds.
\end{lemma}
\begin{proof}
The proof is almost the same as the proof of \cref{lem-algo-2}. The only difference is that here we use \cref{lem:existential-2} instead of \cref{lem:existential-1}.
    We first spend $O(T)$ rounds to gather the entire graph topology of $G[S]$ to $\vstar_S$ for each cluster $S \in \VV$ in the given $(\epsilon, D, T)$-decomposition.
    After that, $\vstar_S$ locally computes an $(\epsilon^\ast, D^\ast, T^\ast)$-decomposition $\VV^\ast$ with $\epsilon^\ast = \frac{\epsilon}{32\alpha}$, $D^\ast = O\left(\frac{1}{\epsilon}\right)$, and $T^\ast = O\left(\frac{\log^5 \Delta \log \frac{1}{\epsilon} +   \log^6 \frac{1}{\epsilon}}{\epsilon^4}\right)$, whose existence is guaranteed by \cref{lem:existential-2}.
    
    The bit string $B_v$, for each $v \in S$, that is needed for running the routing algorithm $\AAA^\ast$ for the  $(\epsilon^\ast, D^\ast, T^\ast)$-decomposition, has two parts. The first part of $B_v$ has length $O\left(\frac{\log^5 \Delta \log \frac{1}{\epsilon} +   \log^6 \frac{1}{\epsilon}}{\epsilon^4}\right) \cdot O(\log n) = T' \cdot O(\log n)$ and is identical for all vertices in $S$, so we can spend $O(T' + D)$ rounds to let $\vstar_S$ broadcast this bit string to all vertices in $S$. The second part of $B_v$  has length $O(\log n)$, so we can spend $O(T)$ rounds to let $\vstar_S$ transmit the bit string $B_v$ and the identifier of the new cluster of $v$  in $\VV^\ast$ to each $v \in S$. 
    
    The union of the $(\epsilon^\ast, D^\ast, T^\ast)$-decompositions for all $S \in \VV$ is a desired $(\epsilon', D', T')$-decomposition because we only create at most $\frac{\epsilon |E|}{32\alpha}$ new inter-cluster edges.
\end{proof}

We are ready to prove the main theorem of this work.
Combining \cref{lem-algo-1,lem-algo-2,lem-algo-3}, we obtain an efficient algorithm for constructing an $(\epsilon, D, T)$-decomposition of any $H$-minor-free graph $G$.

\thmmain*

\begin{proof}
Throughout the proof, we let $\epsilon'$, $D'$, and $T'$ denote the parameters of the current decomposition under consideration. 
We first consider the case $T = 2^{O\left(\log^2 \frac{1}{\epsilon}\right)} \cdot  O\left(\log \Delta \right)$.
We start with the trivial $\left(1, 0, 0\right)$-decomposition where each vertex is a cluster, so we can set $D'= 0$ and $T'=0$.
 Then we iteratively apply the algorithms  of \cref{lem-algo-1} and \cref{lem-algo-2} to the current decomposition, until the parameter $\epsilon'$ of the decomposition becomes at most $\epsilon$. The number of iterations needed is $O\left(\log \frac{1}{\epsilon}\right)$, as $\epsilon'$ is guaranteed to decrease by a factor of $1 - \frac{1}{16\alpha}$ after each iteration of the algorithms of \cref{lem-algo-1} and \cref{lem-algo-2}.
We always have the upper bounds $D' = O\left(\frac{1}{\epsilon}\right)$ and $T' = 2^{O\left(\log^2 \frac{1}{\epsilon}\right)} \cdot  O\left(\log \Delta \right)$ at the end of each iteration, so the decomposition at the end of the last iteration has the desired parameters $(\epsilon, D, T)$.

To analyze the round complexity of the algorithm, let $\epsilon'$, $D' = O\left(\frac{1}{\epsilon'}\right)$, and $T' = 2^{O\left(\log^2 \frac{1}{\epsilon'}\right)} \cdot  O\left(\log \Delta \right)$ denote the parameters at the beginning of one iteration, before running the algorithm  \cref{lem-algo-1}.
The cost of the algorithm of  \cref{lem-algo-1} is \[O\left((D'+1) \log^\ast n + T'\right) =O\left(\frac{\log^\ast n }{\epsilon'}\right) + 2^{O\left(\log^2 \frac{1}{\epsilon'}\right)} \cdot  O\left(\log \Delta \right).\]
After the algorithm of  \cref{lem-algo-1}, the new parameters are $\epsilon'' = \epsilon' \cdot \left(1 - \frac{1}{16\alpha}\right)$, $D'' = 3D'+2$, and $T'' = O\left(\frac{T'+1}{\epsilon}\right)$, so the round complexity of \cref{lem-algo-2} is \[O\left(T'' \cdot \log \frac{1}{\epsilon''} \right) = 
O\left(\frac{T' \cdot \log \frac{1}{\epsilon'} }{\epsilon'}\right)
= 2^{O\left(\log^2 \frac{1}{\epsilon'}\right)} \cdot  O\left(\log \Delta \right).\]
The overall round complexity of one iteration is
 \[O\left(\frac{\log^\ast n }{\epsilon'}\right) + 2^{O\left(\log^2 \frac{1}{\epsilon'}\right)} \cdot  O\left(\log \Delta \right).\]
 Since $\frac{1}{\epsilon'}$ at the beginning of each iteration is exponential in the current iteration number, the round complexity of each iteration also grows at least exponentially, so the overall round complexity is dominated by the round complexity of the last iteration, which is 
\[O\left(\frac{\log^\ast n }{\epsilon}\right) + 2^{O\left(\log^2 \frac{1}{\epsilon}\right)} \cdot  O\left(\log \Delta \right).\]  

The case of $T = O\left(\frac{\log^5 \Delta \log \frac{1}{\epsilon} +   \log^6 \frac{1}{\epsilon}}{\epsilon^4}\right)$ is similar.  We start with the trivial $\left(1, 0, 0\right)$-decomposition, and then we iteratively apply the algorithms  of \cref{lem-algo-1} and \cref{lem-algo-3} to the current decomposition, until the parameter $\epsilon'$ of the decomposition becomes at most $\epsilon$. The number of iterations needed is $O\left(\log \frac{1}{\epsilon}\right)$.
We always have the upper bounds $D' = O\left(\frac{1}{\epsilon}\right)$ and $T' = O\left(\frac{\log^5 \Delta \log \frac{1}{\epsilon} +   \log^6 \frac{1}{\epsilon}}{\epsilon^4}\right)$ at the end of each iteration, so the decomposition at the end of the last iteration has the desired parameters $(\epsilon, D, T)$.

In order to analyze the round complexity of the algorithm, let $\epsilon'$, $D' = O\left(\frac{1}{\epsilon'}\right)$, and $T' = O\left(\frac{\log^5 \Delta \log \frac{1}{\epsilon'} +   \log^6 \frac{1}{\epsilon'}}{{\epsilon'}^4}\right)$ denote the parameters at the beginning of one iteration, before running the algorithm  \cref{lem-algo-1}.
The cost of the algorithm of  \cref{lem-algo-1} is \[O\left((D'+1) \log^\ast n + T'\right) =O\left(\frac{\log^\ast n }{\epsilon'}\right) + O\left(\frac{\log^5 \Delta \log \frac{1}{\epsilon'} +   \log^6 \frac{1}{\epsilon'}}{{\epsilon'}^4}\right).\]
After the algorithm of  \cref{lem-algo-1}, the new parameters are $\epsilon'' = \epsilon' \cdot \left(1 - \frac{1}{16\alpha}\right)$, $D'' = 3D'+2$, and $T'' = O\left(\frac{T'+1}{\epsilon}\right)$, so the round complexity of \cref{lem-algo-3} is \[O\left(T' + T'' + D'\right) = 
O\left(\frac{T'}{\epsilon'}\right)
= O\left(\frac{\log^5 \Delta \log \frac{1}{\epsilon'} +   \log^6 \frac{1}{\epsilon'}}{{\epsilon'}^5}\right).\]
The overall round complexity of one iteration is
 \[O\left(\frac{\log^\ast n }{\epsilon'}\right) + O\left(\frac{\log^5 \Delta \log \frac{1}{\epsilon'} +   \log^6 \frac{1}{\epsilon'}}{{\epsilon'}^5}\right).\]
 Same as the previous case, the overall round complexity is dominated by the round complexity of the last iteration, which is 
\[O\left(\frac{\log^\ast n }{\epsilon}\right) + O\left(\frac{\log^5 \Delta \log \frac{1}{\epsilon} +   \log^6 \frac{1}{\epsilon}}{{\epsilon}^5}\right). \qedhere\] 
\end{proof}

\section{Applications}\label{sect:application}

\cref{lem-algo-main}  immediately leads to an improved algorithm for low-diameter decomposition.

\begin{corollary}\label{cor:ldd}
For any $\epsilon \in \left(0, \frac{1}{2}\right)$, an $(\epsilon, D)$ low-diameter decomposition of any $H$-minor-free graph $G=(V,E)$  with $D = O\left(\frac{1}{\epsilon}\right)$ can be computed deterministically with round complexity
\[R = O\left(\frac{\log^\ast n }{\epsilon}\right) + \min\left\{2^{O\left(\log^2 \frac{1}{\epsilon}\right)} \cdot  O\left(\log \Delta \right),  O\left(\frac{\log^5 \Delta \log \frac{1}{\epsilon} +   \log^6 \frac{1}{\epsilon}}{\epsilon^5}\right)\right\}.\]
\end{corollary}
\begin{proof}
 The lemma follows from \cref{lem-algo-main}, as any $(\epsilon, D, T)$-decomposition is also an $(\epsilon, D)$ low-diameter decomposition. There are two choices of $T$ in \cref{lem-algo-main}. By selecting the one that has the smaller round complexity, we obtain the round complexity $R$ stated int this lemma.
\end{proof}
 
\cref{cor:ldd} is extremely efficient when $\Delta$ or $\frac{1}{\epsilon}$ is small. For bounded-degree graphs, the round complexity of our algorithm  becomes $O\left(\frac{\log^\ast n }{\epsilon}\right) + \poly\left(\frac{1}{\epsilon}\right)$. For constant $\epsilon$, the round complexity is further reduced to $O\left(\log^\ast n\right)$.

Similarly, we may use \cref{lem-algo-main} to obtain expander decompositions matching any given existential bound.

\begin{corollary}\label{cor:expander}
Let $\epsilon \in \left(0, \frac{1}{2}\right)$, and let $G=(V,E)$ be any $H$-minor-free graph with maximum degree $\Delta$. The following expander decompositions can be computed deterministically with round complexity
\[R = O\left(\frac{\log^\ast n }{\epsilon}\right) + \min\left\{2^{O\left(\log^2 \frac{1}{\epsilon}\right)} \cdot  O\left(\log \Delta \right),  O\left(\frac{\log^5 \Delta \log \frac{1}{\epsilon} +   \log^6 \frac{1}{\epsilon}}{\epsilon^5}\right)\right\}.\]
\begin{itemize}
    \item An $\left(\epsilon, \phi \right)$ expander decomposition with $\phi = \Omega\left(\frac{\epsilon}{\log \frac{1}{\epsilon} + \log \Delta}\right)$.
    \item An $\left(\epsilon, \phi, c\right)$ expander decomposition with $\phi = 2^{-O\left(\log^2 \frac{1}{\epsilon}\right)}$ and $c = O\left(\log \frac{1}{\epsilon}\right)$.
\end{itemize}
\end{corollary}
\begin{proof}
By \cref{lem-algo-main}, we may compute an an $\left(\frac{\epsilon}{2}, D, T\right)$-decomposition $\VV$ with $D = O\left(\frac{1}{\epsilon}\right)$ with following round complexities.
\begin{itemize}
    \item For $T = 2^{O\left(\log^2 \frac{1}{\epsilon}\right)} \cdot  O\left(\log \Delta \right)$,  the round complexity is $O\left(\frac{\log^\ast n }{\epsilon}\right) + 2^{O\left(\log^2 \frac{1}{\epsilon}\right)} \cdot  O\left(\log \Delta \right)$.
    \item For $T = O\left(\frac{\log^5 \Delta \log \frac{1}{\epsilon} +   \log^6 \frac{1}{\epsilon}}{\epsilon^4}\right)$, the round complexity is $O\left(\frac{\log^\ast n }{\epsilon}\right) + O\left(\frac{\log^5 \Delta \log \frac{1}{\epsilon} +   \log^6 \frac{1}{\epsilon}}{\epsilon^5}\right)$.   
\end{itemize}
We select $T$ to minimize the round complexity, so the round complexity of the construction is upper bounded by the round complexity $R$ stated in the lemma.

For each cluster $S \in \VV$, we let  $\vstar_S$ locally compute an $\left(\frac{\epsilon}{2}, \phi \right)$ expander decomposition with $\phi = \Omega\left(\frac{\epsilon}{\log \frac{1}{\epsilon} + \log \Delta}\right)$ of $G[S]$, whose existence is guaranteed by~\cref{obs:existence}. After that, $\vstar_S$ can let each $v \in S$ know the cluster that $v$ belongs to in the expander decomposition via the routing algorithm $\AAA$ in $O(T)$ rounds, which is also within the round complexity $R$ stated in the lemma. Taking the union of the $\left(\frac{\epsilon}{2}, \phi \right)$ expander decompositions for all $S \in \VV$ yields a desired $\left(\epsilon, \phi \right)$ expander decomposition of the entire graph $G$.

The construction of an  $\left(\epsilon, \phi, c\right)$ expander decomposition with $\phi = 2^{-O\left(\log^2 \frac{1}{\epsilon}\right)}$ and $c = O\left(\log \frac{1}{\epsilon}\right)$ is similar, so we only discuss the difference. 
For each cluster $S \in \VV$, we let  $\vstar_S$ locally compute an $\left(\frac{\epsilon}{2}, \phi, c \right)$ expander decomposition  of $G[S]$, with  the above $\phi$ and $c$. The existence of such a decomposition is guaranteed by~\cref{lem:existence-variant}. Since the overlap is $c = O\left(\log \frac{1}{\epsilon}\right)$, each vertex $v \in S$ needs $c \cdot O(\deg(v) + \log n)$ bits of information to learn the decomposition. Specifically, suppose $\VV'$ is the $\left(\frac{\epsilon}{2}, \phi, c \right)$ expander decomposition  of $G[S]$. It costs $O(\log n)$ bits for each vertex $v \in S$ to know the part of $\VV'$ that contains $v$. Furthermore, each vertex $v \in S$ belongs to at most $c$ subgraphs in $\{ G_{S'} \, : \, S' \in \VV'\}$.  It costs $c \cdot O(\deg(v) + \log n)$ bits for $v$ to know the edges incident to $v$ in these subgraphs.
The round complexity for $\vstar_S$ to let all vertices in $S$ learn the decomposition is $O(cT) = O\left(T \log \frac{1}{\epsilon}\right)$, which is still within the round complexity $R$ stated in the lemma.
\end{proof}

For the rest of this section, we present further applications of  \cref{lem-algo-main} in approximation algorithms and property testing.

\subsection{Distributed approximation}\label{sect:application-apx}

In this section, we consider applications of \cref{lem-algo-main} to approximation algorithms.
We begin with a simple example of applying \cref{lem-algo-main}.

\begin{corollary}\label{cor:maxcut}
For any $\epsilon \in \left(0, \frac{1}{2}\right)$, a $(1-\epsilon)$-approximate maximum cut of any $H$-minor-free graph $G=(V,E)$   can be computed deterministically with round complexity
\[R = O\left(\frac{\log^\ast n }{\epsilon}\right) + \min\left\{2^{O\left(\log^2 \frac{1}{\epsilon}\right)} \cdot  O\left(\log \Delta \right),  O\left(\frac{\log^5 \Delta \log \frac{1}{\epsilon} +   \log^6 \frac{1}{\epsilon}}{\epsilon^5}\right)\right\}.\]
\end{corollary}
\begin{proof}
We first compute an $\left(\frac{\epsilon}{2}, D, T\right)$-decomposition $\VV$ using \cref{lem-algo-main}. There are two choices of $T$ in \cref{lem-algo-main}, and we select the one that has the smaller round complexity.

For each cluster $S \in \VV$, we let  $\vstar_S$ locally compute a maximum cut $(U_S, S \setminus U_S)$ of $G[S]$. After that, $\vstar_S$ can let each $v \in S$ know whether $v \in U_S$ via the routing algorithm $\AAA$ in $O(T)$ rounds. 
We combine the cuts $(U_S, S \setminus U_S)$ over all $S \in \VV$ into the cut \[\left(\bigcup_{S \in \VV} U_S, V \setminus \bigcup_{S \in \VV} U_S\right),\]
which is guaranteed to have size at least 
\[\OPT - \frac{\epsilon}{2} \cdot |E| \geq (1-\epsilon)\OPT,\] 
where $\OPT \geq |E|/2$ is the size of a maximum cut of $G$.

The overall round complexity is the time needed for constructing an $\left(\frac{\epsilon}{2}, D, T\right)$-decomposition plus the routing time $O(T)$. By \cref{lem-algo-main}, the 
 overall round complexity can be upper bounded by the round complexity $R$ stated in the lemma. 
\end{proof}

\paragraph{Bounded-degree sparsifiers.} Intuitively, the proof idea of \cref{cor:maxcut} allows us to show that $(1\pm \epsilon)$-approximate solutions of many combinatorial optimization problems in $H$-minor-free graphs can be computed in the same round complexity. In the subsequent discussion, we will show that for certain problems, we can further improve the round complexity to  $\poly\left(\frac{1}{\epsilon}\right) \cdot O(\log^\ast n)$, using the \emph{bounded-degree sparsifiers} introduced by Solomon in~\cite{solomon:LIPIcs:2018:8364}.

It was shown in~\cite{solomon:LIPIcs:2018:8364} that for maximum matching, maximum independent set, and minimum vertex cover, there exists a deterministic one-round reduction that reduces the problem of finding a $(1\pm \epsilon)$-approximate solution in a bounded-arboricity graph to the same problem in a subgraph with $\Delta=O\left(\frac{1}{\epsilon}\right)$. Therefore, for these problems, we may focus on the case of $\Delta=O\left(\frac{1}{\epsilon}\right)$, so the proof idea of \cref{cor:maxcut} allows us to find $(1\pm \epsilon)$-approximate solutions in $\poly\left(\frac{1}{\epsilon}\right) \cdot O(\log^\ast n)$ rounds.

 For the sake of completeness, we describe the reduction of~\cite{solomon:LIPIcs:2018:8364} here. Let $G=(V,E)$ be any graph with arboricity at most $\alpha$. Following~\cite{solomon:LIPIcs:2018:8364}, we define:
 \begin{align*}
     V_{\high}^d &= \{v \in V \; | \; \deg(v) \geq d\}, 
&&& G_{\high}^d &= G[V_{\high}^d],\\
 V_{\low}^d &= V \setminus V_{\high}^d,   
 &&& G_{\low}^d &= G[V_{\low}^d].
 \end{align*}
 \begin{itemize}
     \item For the minimum vertex cover problem, $G_{\low}^d$ with $d = O\left(\frac{\alpha}{\epsilon}\right)$ is a sparsifier in the sense that for any $(1+\epsilon)$-approximate minimum vertex cover $C$ of $G_{\low}^d$, $V_{\high}^d \cup C$ is an  $(1+O(\epsilon))$-approximate minimum vertex cover of $G$.
     \item  For the maximum matching problem, let each vertex $v \in V$ mark $\max\{\deg(v), d\}$ arbitrary edges incident to $v$, and let $G_d$ be the subgraph of $G$ induced by the set of all edges that are marked by both endpoints. Taking $d = O\left(\frac{\alpha}{\epsilon}\right)$ ensures that $G_d$ is a sparsifier  in the sense that any $(1-\epsilon)$-approximate maximum matching of $G_d$ is a $(1-O(\epsilon))$-approximate maximum matching of $G$.
     \item For the maximum independent set problem, $G_{\low}^d$ with $d = O\left(\frac{\alpha^2}{\epsilon}\right)$ is a sparsifier in the sense that  any $(1-\epsilon)$-approximate maximum independent set $I$ of $G_{\low}^d$  is an  $(1+O(\epsilon))$-approximate maximum independent set of $G$.
     \end{itemize}

Refer to~\cite{solomon:LIPIcs:2018:8364} for the proof of the above claims.
Combining the proof idea of \cref{cor:maxcut} with the above bounded-degree sparsifiers, we obtain the following results.

\begin{corollary}\label{cor:apx-algo}
For any $\epsilon \in \left(0, \frac{1}{2}\right)$,  $(1\pm\epsilon)$-approximate solutions for the maximum matching and minimum vertex cover problems in any $H$-minor-free graph $G=(V,E)$ can be computed deterministically with round complexity
\[R = O\left(\frac{\log^\ast n }{\epsilon^2}\right) +   O\left(\frac{\log^6 \frac{1}{\epsilon}}{\epsilon^{10}}\right).\]
\end{corollary}
\begin{proof}
In view of the bounded-degree sparsifiers mentioned above, 
we may assume that the maximum degree of $G$ is $\Delta = O\left(\frac{1}{\epsilon}\right)$.
Similar to the algorithms of~\cref{cor:maxcut}, we begin with finding an $\left(\epsilon^\ast, D^\ast, T^\ast\right)$-decomposition $\VV$ using \cref{lem-algo-main} with 
\begin{align*}
\epsilon^\ast &= \frac{\epsilon}{2\Delta-1},\\
D^\ast &= O\left(\frac{1}{\epsilon^\ast}\right) = O\left(\frac{1}{\epsilon^2}\right),\\
T^\ast &= O\left(\frac{\log^5 \Delta \log \frac{1}{\epsilon^\ast} +   \log^6 \frac{1}{\epsilon^\ast}}{{\epsilon^\ast}^4}\right) = O\left(\frac{\log^6 \frac{1}{\epsilon}}{\epsilon^8}\right).
\end{align*}
 The round complexity of the construction of $\VV$ is
\[O\left(\frac{\log^\ast n }{\epsilon^\ast}\right) + O\left(\frac{\log^5 \Delta \log \frac{1}{\epsilon^\ast} +   \log^6 \frac{1}{\epsilon^\ast}}{{\epsilon^\ast}^5}\right)
=  O\left(\frac{\log^\ast n }{\epsilon^2}\right) +   O\left(\frac{\log^{6} \frac{1}{\epsilon}}{\epsilon^{10}}\right),\]
 which is within the round complexity $R$ stated in the lemma.

\paragraph{Maximum matching.}
We begin with the maximum matching problem.
Any maximal matching has size at least $\frac{|E|}{2\Delta - 1}$, as each edge is adjacent to at most $2\Delta-2$ edges. Therefore, we have $\OPT \geq \frac{|E|}{2\Delta - 1}$, where $\OPT$ is the size of a maximum matching of $G$.
Consider the graph $G'$ resulting from removing all the inter-cluster edges of $\VV$. The size of a maximum matching of $G'$ is at least 
\[\OPT - \epsilon^\ast|E|
= \OPT - \frac{\epsilon |E|}{2\Delta-1}
\geq \OPT(1 - \epsilon).
\]

As a matching of $G'$ is also a matching of $G$, to find a $(1 - \epsilon)$-approximate maximum matching, it suffices to compute a maximum matching of $G'$, and this can be done by taking the union of a maximum matching of $G[S]$, for all clusters $S \in \VV$. 
This can be done in $O(T^\ast)$ rounds by gathering the entire graph topology of $G[S]$ to the leader $\vstar_S$ using the routing algorithm $\AAA$. The round complexity is within the round complexity $R$ stated in the lemma.

\paragraph{Minimum vertex cover.} Any vertex cover has size at least $\frac{|E|}{\Delta}$, as each vertex can cover at most $\Delta$ edges, so we have $\OPT \geq \frac{|E|}{\Delta}$, where $\OPT$ is the size of a minimum vertex cover of $G$.
Consider the graph $G'$ resulting from removing all the inter-cluster edges of $\VV$. 
Let $C'$ be any minimum vertex cover of $G'$.
Let $C$ be the vertex cover of $G$ constructed as follows. Start from $C'$. For each inter-cluster edge $e=\{u,v\}$, add either $u$ or $v$ to the vertex cover.
The size of $C$ is at most
\[
|C'| + \epsilon^\ast|E|
\leq \OPT + \epsilon^\ast|E|
= \OPT + \frac{\epsilon |E|}{2\Delta-1}
\leq \OPT + \frac{\epsilon |E|}{\Delta}
= \OPT(1 + \epsilon).
\]

In view of the above discussion, to find a $(1 + \epsilon)$-approximate minimum vertex cover of $G$, it suffices to compute a minimum vertex cover of $G'$, and this can be done by taking the union of a minimum vertex cover of $G[S]$, for all clusters $S \in \VV$. 
This can be done in $O(T^\ast)$ rounds by gathering the entire graph topology of $G[S]$ to the leader $\vstar_S$ using the routing algorithm $\AAA$. The round complexity is within the round complexity $R$ stated in the lemma.
\end{proof}

As the size of a maximum independent set in any $H$-minor-free graph $G=(V,E)$ is $\Theta(|E|)$, we obtain a better round complexity for $(1-\epsilon)$-approximate maximum independent set. The proof of the following corollary is similar to the proof of \cite[Theorem 1.2]{10.1145/3519270.3538423}.

\begin{corollary}\label{cor:apx-is}
For any $\epsilon \in \left(0, \frac{1}{2}\right)$, a $(1-\epsilon)$-approximate   maximum independent set of any $H$-minor-free graph $G=(V,E)$ can be computed deterministically with round complexity
\[R = O\left(\frac{\log^\ast n }{\epsilon}\right) +   O\left(\frac{\log^6 \frac{1}{\epsilon}}{\epsilon^5}\right).\]
\end{corollary}
\begin{proof}
In view of the bounded-degree sparsifiers mentioned above, 
we may assume that the maximum degree of $G$ is $\Delta = O\left(\frac{1}{\epsilon}\right)$.
Similar to the algorithms of~\cref{cor:maxcut,cor:apx-algo}, we begin with finding an $\left(\epsilon^\ast, D^\ast, T^\ast\right)$-decomposition $\VV$ using \cref{lem-algo-main} with 
\begin{align*}
\epsilon^\ast &= \frac{\epsilon}{\alpha(2\alpha-1)},\\
D^\ast &= O\left(\frac{1}{\epsilon^\ast}\right) = O\left(\frac{1}{\epsilon}\right),\\
T^\ast &= O\left(\frac{\log^5 \Delta \log \frac{1}{\epsilon^\ast} +   \log^6 \frac{1}{\epsilon^\ast}}{{\epsilon^\ast}^4}\right) = O\left(\frac{\log^6 \frac{1}{\epsilon}}{\epsilon^4}\right).
\end{align*}
 The round complexity of the construction of $\VV$ is
\[O\left(\frac{\log^\ast n }{\epsilon^\ast}\right) + O\left(\frac{\log^5 \Delta \log \frac{1}{\epsilon^\ast} +   \log^6 \frac{1}{\epsilon^\ast}}{{\epsilon^\ast}^5}\right)
=  O\left(\frac{\log^\ast n }{\epsilon}\right) +   O\left(\frac{\log^6 \frac{1}{\epsilon}}{\epsilon^5}\right),\]
 which is within the round complexity $R$ stated in the lemma.

 \paragraph{Algorithm.} Given $\VV$, we present our algorithm for constructing an independent set $I$ of $G$.
Let $G'$ be the result of removing all the inter-cluster edges of $G$. We first compute a maximum independent set $I'$ of $G'$, and then $I$ is constructed as follows. Start with the independent set $I'$, and then for each inter-cluster edge $e=\{u,v\}$ such that $\{u,v\}\subseteq I'$, remove either $u$ or $v$ from the independent set. 

A maximum independent set $I'$ of $G'$ can be computed by taking the union of a maximum independent set of $G[S]$, for all clusters $S \in \VV$. 
This can be done in $O(T^\ast)$ rounds by gathering the entire graph topology of $G[S]$ to the leader $\vstar_S$ using the routing algorithm $\AAA$. The round complexity of the construction of $I'$ and $I$ is within the round complexity $R$ stated in the lemma.

\paragraph{Analysis.}
For the rest of the proof, we show that $|I| \geq (1-\epsilon)\OPT$, where $\OPT$ is the size of a maximum independent set of $G$. 
Observe that $2\alpha \geq \frac{ 2|E|}{|V|-1} >   \frac{2|E|}{|V|}$, so any graph with arboricity at most $\alpha$ has minimum degree at most $2\alpha -1$. Therefore, $\OPT \geq \frac{|V|}{2\alpha-1} \geq \frac{|E|}{\alpha(2\alpha-1)}$.
Using the facts $|I'| \geq \OPT$ and $|I| \geq |I'| - \epsilon^\ast|E|$, we lower bound $|I|$ as follows:
\[|I| \geq |I'| - \epsilon^\ast|E| \geq \OPT - \epsilon^\ast|E|
= \OPT - \frac{\epsilon |E|}{\alpha(2\alpha-1)}
\geq \OPT(1-\epsilon),\]
so our algorithm outputs  a $(1-\epsilon)$-approximate   maximum independent set of $G$.
\end{proof}

\paragraph{Lower bound.} The round complexity of \cref{cor:apx-is} is nearly optimal, up to an additive $\poly\left(\frac{1}{\epsilon}\right)$ term, in view of the $\Omega\left(\frac{\log^\ast n}{\epsilon}\right)$ lower bound for finding a $(1-\epsilon)$-approximate solution for the maximum independent set problem by Lenzen and  Wattenhofer~\cite{LenzenW08}, which holds even for paths and cycles in the deterministic $\LOCAL$ model. 

\begin{theorem}[\cite{LenzenW08}]\label{thm-lb-apx}
 Any deterministic algorithm that finds a $(1-\epsilon)$-approximate solution for the maximum independent set problem requires $\Omega\left(\frac{\log^\ast n}{\epsilon}\right)$ rounds for  paths and cycles in the $\LOCAL$ model. 
\end{theorem}

As the lower bound of \cref{thm-lb-apx} holds 
 for paths and cycles, with a straightforward reduction, the same $\Omega\left(\frac{\log^\ast n}{\epsilon}\right)$ lower bound also applies to maximum matching and minimum vertex cover.
  
\subsection{Distributed property testing}\label{sect:application-test}

In this section, we show that \cref{lem-algo-main} yields an efficient property testing algorithm for all additive and minor-closed graph properties. 

If a property $\mathcal{P}$ is minor-closed, then  all  $G \in \mathcal{P}$ are $H$-minor-free, for any choice of  $H \notin \mathcal{P}$. To put it another way, for any minor-closed graph property $\mathcal{P}$ that is not the set of all graphs, there exists a fixed graph $H$ such that all $G \in \mathcal{P}$ are $H$-minor-free, so the algorithm of \cref{lem-algo-main} will work correctly in $G$ with this choice of $H$.

Intuitively, if we are given an $\left(\epsilon, D, T\right)$-decomposition $\VV$, which can be computed using \cref{lem-algo-main}, then we should be able to design a property testing algorithm that costs $O(T)$ rounds, as we can afford to ignore all the inter-cluster edges. 
For this approach to work, we require that the graph property $\PP$ under consideration is additive, as 
we need to make sure that  $G[S] \in \PP$ for all $S \in \VV$ implies that their disjoint union also has property $\PP$. 
 
\paragraph{Error detection.}
There is still one caveat in applying the above approach, that is, the correctness \cref{lem-algo-main} relies on the assumption that $G$ is $H$-minor-free, but in the context of distributed property testing, $G$ can be any graph, so we need to make sure that whenever \cref{lem-algo-main} produces an incorrect output due to the fact that $G$ is not $H$-minor-free, at least one vertex $v \in V$ can detect it and output $\reject$.
Specifically, there are three steps of the algorithm of \cref{lem-algo-main} that rely on the assumption that  $G$ is $H$-minor-free. 
\begin{itemize}
    \item Whenever we apply the routing algorithms of \cref{lem:gathering-1,lem:gathering-2,lem:gathering-3} to a subgraph $G'=(V',E')$ that is an $\phi$-expander  during the algorithm of \cref{lem-algo-main}, we require that the maximum degree $\Delta'$ of $G'$ is at least $\Omega(\phi^2|E'|)$, for the routing algorithms to be efficient. This bound holds for any $H$-minor-free graph due to~\cref{lem:separator}. To check whether this bound holds, we can simply check in $G'$ whether its maximum degree $\Delta'$ is sufficiently large. If $\Delta'$ is too small, then the vertices in $G'$ know that $G'$ is not $H$-minor-free, which implies that the original graph $G$ is also not $H$-minor-free, so the vertices in $G'$ can safely output $\reject$.
    \item Whenever we apply  \cref{lem:LDD-basic} to find a low-diameter decomposition during the algorithm of \cref{lem-algo-main}, we need to ensure that such a low-diameter decomposition stated in \cref{lem:LDD-basic} exists, and this requires the subgraph $G'=(V',E')$ under consideration to be $H$-minor-free. Observe that we never construct the decomposition of \cref{lem:LDD-basic} in a distributed manner. We only apply \cref{lem:LDD-basic} when we have gathered the entire topology of some graph $G'$ to a vertex $v$, and then $v$ uses its local computation power to find a decomposition of \cref{lem:LDD-basic}. Therefore, if $G'$ is not $H$-minor-free, then $v$ is able to detect that, in which case $v$ can output $\reject$.
    \item Whenever we apply the heavy-stars algorithm during the algorithm of \cref{lem-algo-main}, in order to make sure that sufficiently many inter-cluster edges are captured in the stars in \cref{lem:heavystar-weight}, we require that the cluster graph  under consideration has arboricity  at most $\alpha$, where $\alpha=O(1)$ can be any given arboricity upper bound for $H$-minor-free graphs. In the following discussion, we discuss how we can choose $\alpha$ so that an error can be detected.
\end{itemize}

\paragraph{Forests decomposition.} 
Let $\alpha_0 = O(1)$ be a given arboricity upper bound for $H$-minor-free graphs.
The issue mentioned above can be resolved using the \emph{forests decomposition} algorithm of Barenboim and Elkin~\cite{BE10}, which allows us to distinguish between graphs with arboricity at most $\alpha_0$ and graphs  with arboricity greater than $3 \alpha_0$ in the following sense.
\begin{itemize}
    \item If the underlying graph has arboricity at most $\alpha_0$, then no vertex outputs $\reject$.
    \item If the underlying graph has arboricity greater than $3\alpha_0$, then some vertex outputs $\reject$.
\end{itemize}
Therefore, setting $\alpha = 3\alpha_0$ and running this algorithm on the cluster graph $G'$ resolves the issue mentioned above.

Specifically, it was shown in~\cite{BE10}  that for any graph $G'=(V',E')$ with arboricity at most $\alpha_0$, there is an $O(\log n)$-round deterministic algorithm $\AAA$ that computes an acyclic orientation of the edges in $E'$ in such a way that the outdegree of each vertex $v \in V'$ is at most $3 \alpha_0$. As the orientation is acyclic,  we may use such an orientation to compute a decomposition of $E'$ into $3 \alpha_0$ forests, certifying that the arboricity of $G'$ is at most $3 \alpha_0$.

If we run the algorithm of Barenboim and Elkin on a graph $G'=(V',E')$ whose arboricity exceeds $\alpha_0$, the algorithm still orients a subset of $E'$ such that the resulting orientation is acyclic and the outdegree of every vertex $v \in V'$ is at most $3\alpha_0$. Consequently, we let a vertex output $\reject$ if it has an incident edge that remains unoriented; in this case, the arboricity of $G'$ must exceed $\alpha_0$, and hence $G'$ is not $H$-minor-free.

Moreover, if the arboricity of $G'$ exceeds $\alpha = 3\alpha_0$, then at least one edge is guaranteed to remain unoriented, implying that some vertex outputs $\reject$. Therefore, if no vertex outputs $\reject$, we can conclude that $G'$ has arboricity at most $\alpha = 3\alpha_0$, as required.

The algorithm of Barenboim and Elkin in a graph $G'=(V',E')$ works as follows. For $i = 1, 2, \ldots, O(\log n)$, define $U_i$ as the set of vertices $v \in V' \setminus \bigcup_{j=1}^{i-1} U_j$ whose degree in $V' \setminus \bigcup_{j=1}^{i-1} U_j$ is at most $3 \alpha_0$. For each edge $e=\{u,v\}$ such that $u \in U_i$ and $v \in U_j$ with $i < j$, orient the edge in the direction $u \rightarrow v$.  For each edge $e=\{u,v\}$ such that both $u$ and $v$ belong to the same set $U_i$ and $\ID(v) > \ID(u)$, orient the edge in the direction $u \rightarrow v$.

\paragraph{Implementation.}
We show how to efficiently implement the algorithm of Barenboim and Elkin on the cluster graph, and how this yields an error-detection mechanism that adds only an $O\left(\frac{\log n}{\epsilon}\right)$ additive term to the construction time of an $(\epsilon, D, T)$-decomposition in \cref{lem-algo-main}.

We first argue that the algorithm can be executed on the cluster graph associated with an $(\epsilon', D', T')$-decomposition in $O(D' \log n)$ rounds, after a preprocessing step that takes $O(T')$ rounds.

In the preprocessing step, for each cluster $S$, every vertex $v \in S$ sends its list of neighboring clusters to the leader $\vstar_S$ using the routing algorithm in $T'$ rounds. Then, for each neighboring cluster $S'$, the leader $\vstar_S$ designates exactly one vertex in $S$ that is adjacent to $S'$ to be responsible for handling this adjacency. This notification can be completed in another $T'$ rounds by running the routing algorithm in reverse.

After preprocessing, each iteration of the Barenboim--Elkin algorithm can be implemented in $O(D')$ rounds. For each cluster $S$, it suffices to determine whether its degree in $V' \setminus \bigcup_{j=1}^{i-1} U_j$ is at most $3\alpha_0$. This can be done using an $O(D')$-round aggregation via a BFS tree: each vertex $v \in S$ reports the number of neighboring clusters in $V' \setminus \bigcup_{j=1}^{i-1} U_j$ for which it is responsible, and the root of the BFS tree aggregates these values.

The $(\epsilon, D, T)$-decomposition algorithm of \cref{lem-algo-main} consists of $O\left(\log \frac{1}{\epsilon}\right)$ iterations. In each iteration, we run the heavy-stars algorithm on the cluster graph associated with some $(\epsilon', D', T')$-decomposition. Since the construction algorithm of \cref{lem-algo-main} already invokes the routing algorithm for this clustering, the preprocessing cost of $O(T')$ can be absorbed into the overall construction time without affecting the asymptotic complexity.

For the execution of the Barenboim–Elkin algorithm, the total cost is bounded by $O\left(\frac{\log n}{\epsilon}\right)$. This follows from the fact that $D' = O\left(\frac{1}{\epsilon'}\right)$, and $\frac{1}{\epsilon'}$ grows exponentially with the iteration number until it reaches $\frac{1}{\epsilon}$. Therefore, the overall round complexity is dominated by the final iteration, which takes $O\left(\frac{\log n}{\epsilon}\right)$ rounds.

We now state our result for distributed property testing. We emphasize that the algorithm of \Cref{cor:testing} applies to \emph{general} graphs: the minor-closed assumption applies to the property $\PP$, not to the underlying network.

\begin{corollary}\label{cor:testing}
Let $\PP$ be any additive minor-closed graph property.
For any $\epsilon \in \left(0, \frac{1}{2}\right)$, there is a deterministic distributed property testing algorithm for  $\PP$ whose round complexity is
\[R = O\left(\frac{\log n }{\epsilon}\right) + \min\left\{2^{O\left(\log^2 \frac{1}{\epsilon}\right)} \cdot  O\left(\log \Delta \right),  O\left(\frac{\log^5 \Delta \log \frac{1}{\epsilon} +   \log^6 \frac{1}{\epsilon}}{\epsilon^5}\right)\right\}.\]
\end{corollary}
\begin{proof}
Let $G=(V,E)$ be the input graph.
If $\PP$ is the set of all graphs, then we can simply let all $v \in V$ output $\accept$ without any communication.
Otherwise, we pick any $H \notin \PP$, so all graphs that have property $\PP$ are $H$-minor-free.

\paragraph{Decomposition.}
We compute
an $\left(\epsilon^\ast, D^\ast, T^\ast\right)$-decomposition $\VV$ using \cref{lem-algo-main} with 
\begin{align*}
\epsilon^\ast &= \frac{\epsilon}{2},\\
D^\ast &= O\left(\frac{1}{\epsilon^\ast}\right) = O\left(\frac{1}{\epsilon}\right),\\
T^\ast &= O\left(\frac{\log^5 \Delta \log \frac{1}{\epsilon^\ast} +   \log^6 \frac{1}{\epsilon^\ast}}{{\epsilon^\ast}^4}\right) = O\left(\frac{\log^6 \frac{1}{\epsilon}}{\epsilon^4}\right).
\end{align*}
In view of the above discussion, we can modify the algorithm of \cref{lem-algo-main} to accommodate the possibility that $G$ can be any graph to ensure that the output of the algorithm satisfies either one of the following conditions.
\begin{itemize}
    \item $G$ is not $H$-minor-free and at least one vertex $v \in V$ output $\reject$.
    \item The computed $\left(\epsilon^\ast, D^\ast, T^\ast\right)$-decomposition $\VV$ satisfies all the requirements stated in \cref{lem-algo-main}.
\end{itemize}
The modification consists of some error detection mechanisms ensuring that some vertices will output $\reject$ whenever the  algorithm  of \cref{lem-algo-main} runs incorrectly because $G$ is not $H$-minor-free. As discussed earlier, the cost of such an error detection is an additive $O\left(\frac{\log n }{\epsilon}\right)$ term in the round complexity, as we need to run the forests decomposition algorithm of Barenboim and Elkin~\cite{BE10}.

\paragraph{Local computation.}
After computing the  $\left(\epsilon^\ast, D^\ast, T^\ast\right)$-decomposition $\VV$, we let the leader $\vstar_S$ of each cluster $S$ gather the entire graph topology of $G[S]$, and then $\vstar_S$ locally checks whether $G[S] \in \PP$ and announces the result to all vertices in $S$. This step can be done in $O(T^\ast)$ rounds using the routing algorithm $\AAA$ associated with the $\left(\epsilon^\ast, D^\ast, T^\ast\right)$-decomposition $\VV$.
For each cluster $S \in \VV$, each $v \in S$ decides its output as follows. If $v$  already output $\reject$, then $v$ does not change its decision. Otherwise, $v$ outputs $\accept$ if  $G[S] \in \PP$ and outputs $\reject$ if $G[S] \notin \PP$.

\paragraph{Round complexity.}
If the $\left(\epsilon^\ast, D^\ast, T^\ast\right)$-decomposition $\VV$ is correctly computed, then the overall round complexity of our algorithm is within the round complexity $R$ stated in the corollary.
There is however a possibility that the computation exceeds $R$ rounds due to an error in the  $\left(\epsilon^\ast, D^\ast, T^\ast\right)$-decomposition $\VV$, in which case we let all vertices that have not stopped by the given time limit $R$ to stop and output $\reject$.

\paragraph{Correctness.}
For the rest of the proof, we show that the algorithm is correct. 
Consider the case where $G$ has property $\PP$. In this case, $G$ is $H$-minor-free, so the computed $\left(\epsilon^\ast, D^\ast, T^\ast\right)$-decomposition $\VV$ satisfies all the requirements stated in \cref{lem-algo-main}. Therefore, no one outputs $\reject$ during the computation of $\VV$ due to error detection, and no one outputs $\reject$ due to exceeding the time limit $R$. Since $\PP$ is minor-closed, the subgraph $G[S]$ induced by $S$, for each $S \in \VV$, also has property $\PP$, so all vertices output $\accept$, as required.

Consider the case where $G$ is $\epsilon$-far from having property $\PP$. In this case, $G$ is not $H$-minor-free. Suppose all $v \in V$ output $\accept$. Then $G[S] \in \PP$ for all clusters $S \in \VV$. Since $\PP$ is additive, the disjoint union $G^\ast$ of $G[S]$ for all clusters $S \in \VV$ still has property $\PP$. Since no one outputs $\reject$, the decomposition $\VV$ is correct, meaning that the number of inter-cluster edges is at most $\epsilon^\ast |E| \leq \frac{\epsilon|E|}{2}$, so  $G^\ast$  can be obtained by deleting at most $\frac{\epsilon|E|}{2}$ edges in $G$. contradicting the assumption that $G$ is $\epsilon$-far from having property $\PP$. Therefore, at least one vertex $v \in V$ must output $\reject$.
\end{proof}

\paragraph{Lower bound.} 
 There is an $\Omega\left(\frac{\log n }{\epsilon}\right)$ lower bound~\cite{levi2021property} for distributed property testing for a wide range of minor-closed graph properties $\PP$.  Although the lower bound stated in~\cite{levi2021property} only considers the case where $\epsilon$ is some constant, it is straightforward to extend the lower bound to smaller $\epsilon$ by subdividing edges into paths of length $O\left(\frac{1}{\epsilon}\right)$.
 
 \begin{theorem}[\cite{levi2021property}]\label{thm:lb-testing}
 Let $\PP$ be any   graph property meeting the following  conditions.
 \begin{itemize}
 \item $\PP$ is minor-closed.
     \item $\PP$  contains all forests.
     \item $\PP$  is not the set of all graphs.
 \end{itemize}
 There is a number $\epsilon_0 \in (0,1)$ depending only on $\PP$ such that for any $\epsilon \in (0,\epsilon_0]$, any deterministic distributed property testing algorithm for $\PP$ requires $\Omega\left(\frac{\log n }{\epsilon}\right)$ rounds in the $\LOCAL$ model. 
 \end{theorem}
 \begin{proof}
     We fix $H$ to be any graph that is not in $\PP$. Let $t=O(1)$ be the number of vertices in $H$. We set $\epsilon_0 = \frac{1}{100t^2}$. It was shown in~\cite[Claim 11]{levi2021property} that there exists an infinite family of graphs $\mathcal{G}$ meeting the following two conditions.
     \begin{itemize}
         \item Each $G \in \GGG$ is $\epsilon_0$-far from $K_t$-minor-freeness.
         \item Each cycle in each $G \in \GGG$ has length   $\Omega(\log n)$, where $n$ is the number of vertices in $G$. 
     \end{itemize}
     
     The first condition implies that each $G \in \mathcal{G}$ is  $\epsilon_0$-far from $H$-minor-freeness, so $G$  is also $\epsilon_0$-far from having property $\PP$.
     The second condition implies that if we run any $o(\log n)$-round deterministic distributed algorithm $\AAA$ for property testing $\PP$ in a sufficiently large graph $G=(V,E) \in \GGG$, then each vertex $v \in V$ cannot distinguish between the following two cases.
     \begin{itemize}
         \item The underlying graph is in $\GGG$, in which case the  graph is $\epsilon_0$-far from having property $\PP$.
         \item The underlying graph is a tree, in which case the  graph has property $\PP$.
     \end{itemize}
     If $v$ outputs $\accept$, then $v$ makes an error when the underlying graph is in $\GGG$. If $v$ outputs $\reject$, then $v$ makes an error when the underlying graph is a tree. 
     Therefore, such a property testing algorithm  $\AAA$ cannot be correct, so we obtain an $\Omega(\log n)$ lower bound for $\epsilon = \epsilon_0$.

     \paragraph{Extension.} To extend the lower bound to any $\epsilon \in (0,\epsilon_0)$, we subdivide each edge in each graph  $G=(V,E) \in \GGG$ into a path of length $\left\lfloor\frac{\epsilon_0}{\epsilon}\right\rfloor$. This modification increases the number of edges $|E|$ by a factor of $\left\lfloor\frac{\epsilon_0}{\epsilon}\right\rfloor$. As $\frac{\epsilon_0}{\left\lfloor\frac{\epsilon_0}{\epsilon}\right\rfloor} \geq \epsilon$, after the modification, each $G \in \GGG$ is $\epsilon$-far from $K_t$-minor-freeness, so $G$ is also $\epsilon$-far from  having property $\PP$.
    The subdivision also implies that each cycle in each $G \in \GGG$ has length   $\left\lfloor\frac{\epsilon_0}{\epsilon}\right\rfloor \cdot \Omega(\log n) = \Omega\left(\frac{\log n}{\epsilon}\right)$, where $n$ is the number of vertices in $G$. Therefore, the above lower bound argument implies that any $o\left(\frac{\log n}{\epsilon}\right)$-round deterministic property testing algorithm for $\PP$ cannot be correct, so we obtain the desired $\Omega\left(\frac{\log n}{\epsilon}\right)$ lower bound.

     \paragraph{Remark.} Although the construction of the graph family $\GGG$ in~\cite{levi2021property}  involves graphs that do not have a constant maximum degree, as noted in~\cite{levi2021property}, such a construction can be modified in such a way that involve only graphs with maximum degree $\Delta = O(1)$ via the approach of~\cite{censor2019fast}. Hence the lower bound of this lemma holds even for bounded-degree graphs.
 \end{proof}

 For example, since the set of planar graphs includes all forests, \cref{thm:lb-testing} implies that any deterministic distributed property testing algorithm for planarity requires $\Omega\left(\frac{\log n }{\epsilon}\right)$ rounds.
 
 The lower bound of \cref{thm:lb-testing}  implies that \cref{cor:testing} is nearly optimal or optimal when $\frac{1}{\epsilon}$ or $\Delta$ is small. If $\epsilon$ is a constant,   then the algorithm of  \cref{cor:testing} costs only $O(\log n)$ rounds, which is optimal.  If $\Delta$ is a constant, then the algorithm of  \cref{cor:testing} costs only $O\left(\frac{\log n }{\epsilon}\right) +  \poly\left(
 \frac{1}{\epsilon}\right)$ rounds, which is nearly optimal, up to an additive $\poly\left(
 \frac{1}{\epsilon}\right)$ term.

\section{Conclusions and open questions}\label{sect:conclusions}
In this work, we design efficient deterministic distributed algorithms for computing an $(\epsilon, D, T)$-decomposition in $H$-minor-free networks, improving upon the results of~\cite{10.1145/3519270.3538423}, which are based on expander decomposition and routing techniques from~\cite{ChangS20}.

From a technical perspective, our decomposition is constructed in a bottom-up manner through iterative cluster merging using the heavy-stars algorithm of~\cite{czygrinow2008fast}. In each iteration, we incorporate an information-gathering procedure to improve the quality of the decomposition. This procedure builds on a load-balancing algorithm of~\cite{GhoshLMMPRRTZ99} and a derandomization of random walks using limited independence.

Our new decomposition algorithm leads to improved bounds for approximation algorithms and property testing, and in several cases achieves optimal or near-optimal upper bounds.

\paragraph{Open questions.}
A natural open question is to further improve the construction time and routing time $T$ for $(\epsilon, D, T)$-decomposition in $H$-minor-free networks. Since the way our decomposition algorithm works is to use an information-gathering algorithm to improve the quality of the decomposition, if the \emph{existential} bounds for expander decompositions in $H$-minor-free networks in~\cref{obs:existence,lem:existence-variant} are improved, then we automatically obtain an improved \emph{distributed} algorithm for $(\epsilon, D, T)$-decomposition. What is the optimal conductance bound for expander decompositions of $H$-minor-free graphs?

For any combinatorial optimization problem that admits a bounded-degree sparsifier as defined in~\cite{solomon:LIPIcs:2018:8364}, $(1\pm \epsilon)$-approximate solutions of the problem in $H$-minor-free graphs can be computed in $\poly\left(\frac{1}{\epsilon}\right) \cdot O(\log^\ast n)$ rounds in $\CONGEST$ via our approach. Apart from maximum matching, maximum independent set, and minimum vertex cover, which other combinatorial optimization problems also admit bounded-degree sparsifiers?

The property testing algorithms in this paper and in~\cite{10.1145/3519270.3538423} only apply to minor-closed graph properties that are \emph{additive}. To what extent the underlying approach of these algorithms can be extended to minor-closed graph properties that are not additive? Notably, it remains unclear whether it is possible to design a sublinear-round distributed property testing algorithm for the property $\mathcal{P} = \text{the class of graphs that can be embedded on a torus}$.

\bibliographystyle{alpha}
\bibliography{references}
 
\appendix

\end{document}